\documentclass[11pt]{article}
\usepackage[margin=1.05in]{geometry}
\usepackage{amsmath,amssymb,amsthm,mathtools}
\usepackage{tabularx}
\usepackage{booktabs,graphicx,enumitem,microtype,placeins}
\usepackage{bm}
\usepackage{xcolor}
\usepackage{tikz}
\usetikzlibrary{arrows.meta,positioning,calc,decorations.pathmorphing,decorations.pathreplacing,backgrounds,fit}
\usepackage[colorlinks=true,linkcolor=blue!60!black,citecolor=blue!60!black,urlcolor=blue!60!black]{hyperref}
\usepackage[capitalize,nameinlink]{cleveref}

\newtheorem{theorem}{Theorem}[section]
\newtheorem{proposition}[theorem]{Proposition}
\newtheorem{lemma}[theorem]{Lemma}
\newtheorem{corollary}[theorem]{Corollary}
\theoremstyle{definition}
\newtheorem{definition}[theorem]{Definition}

\theoremstyle{remark}
\newtheorem{remark}[theorem]{Remark}

\newcommand{\R}{\mathbb{R}}
\newcommand{\E}{\mathbb{E}}

\newcommand{\F}{\mathcal{F}}
\newcommand{\G}{\mathcal{G}}
\newcommand{\M}{\mathcal{M}}
\newcommand{\W}{\mathcal{W}}
\newcommand{\LL}{\mathcal{L}}

\newcommand{\one}{\mathbf{1}}
\newcommand{\tr}{\operatorname{tr}}
\newcommand{\diag}{\operatorname{diag}}

\newcommand{\Dstar}{\mathfrak{D}^{*}}
\newcommand{\Dset}{\mathfrak{D}}
\newcommand{\bb}{\mathfrak{b}}

\title{Dynamic Causal Portfolio Choice:\\ Hedging the Rotation of the Common-Driver Manifold}
\author{Alejandro Rodr\'iguez Dom\'inguez\thanks{Quantitative Analysis and Artificial Intelligence Department, Miralta Finance Bank S.A., Madrid, Spain; and Department of Computer Science, University of Reading, Reading, United Kingdom. Email: \texttt{arodriguez@miraltabank.com}. A reproducibility package accompanying this paper is available online.}}
\date{\today}

\begin{document}
\maketitle

\begin{abstract}
When a portfolio is conditioned on a minimal set of observable drivers under which its assets become mutually independent over the investment horizon, the dynamic investment problem acquires a distinctive geometric structure. We study continuous-time portfolio choice in this setting. The conditioning representation, rather than the asset vector, becomes the natural state of the problem, and it moves: the sensitivity of returns to the drivers depends on the state, the conditioning set may itself change over time, and the induced information geometry both rotates and, at discrete instants, jumps. The optimal policy separates into a static component that allocates along the conditioning geometry at each instant and an intertemporal component that hedges the predictable motion of that geometry, a first-order effect in the model rather than a refinement, placing the coordinates of the information geometry in the role played by exogenous state variables in classical intertemporal asset pricing. Because the problem is organized by the drivers, its computational cost is governed by their number rather than by the number of assets. Changes in the conditioning set generate a risk that continuous trading cannot span, so the market is incomplete in the direction of its own geometry. The analysis is carried out in a controlled diffusion model, and the resulting structure is illustrated on synthetic economies designed to isolate each mechanism.
\end{abstract}



\section{Introduction}\label{sec:intro}

Static mean--variance choice conditioned on a causal separator has a closed-form solution: the projected Markowitz portfolio, in which the precision matrix is replaced by its projection onto the subspace of exposures compatible with the geometry of the conditioning drivers. That theory is developed in a companion paper \cite{RD2025Static}, where separation is taken as a probabilistic primitive at the level of $\sigma$-algebras and its consequences for a single-period program are characterized. The present paper asks what happens when the horizon is long enough that the separator, the response maps and the induced geometry move, so that the investor faces an intertemporal problem rather than a sequence of unrelated static ones.

We state the scope of the results at the outset. The dynamic control results of this paper rest on \emph{conditional-factorization} hypotheses, namely a low-rank conditional covariance, a diagonal idiosyncratic block, and state-dependent loadings, and are, strictly, statements about that conditional structure. Their identification with a \emph{causal} separator, the guarantee that the conditioning set is intervention invariant rather than merely predictive, is established in the static companion \cite{RD2025Static}; the single consequence the dynamic theory uses, invariance of the priced geometry under environment interventions, is stated and proved here from an explicit structural hypothesis, so the development is self-contained given that hypothesis. A reader may take the development as a conditional-factor dynamic portfolio theory; the causal reading is the additional content the companion supplies. With that understood, the dynamic structure is governed by the same object that organizes the static theory, the common manifold, now treated as a state variable in its own right. Three features distinguish the dynamic problem. First, the optimal policy is no longer purely myopic: it acquires a hedging term against the predictable motion of the representation, which is Merton's intertemporal hedging demand \cite{Merton1971,Merton1973} with the coordinates of the information geometry playing the role of the priced state. Second, because the apparatus is organized by the drivers, the intertemporal problem inherits the dimension reduction of the static one: the governing equation is posed in the driver state, whose dimension is the number of drivers rather than the number of assets. Third, the separator itself can switch, and when it does the value process jumps; in a continuous-asset market the jump is unspanned by any traded strategy, so the market is incomplete in the direction of its own geometry, a statement conditional on the traded universe.

This paper is theoretical, and its experiments are synthetic by design: the contribution is the mathematical content of the dynamic representation, the hedging decomposition, the geometry of switches, and the transport connection, established against a known generating process in which each claim is falsifiable. The empirical, out-of-sample side of the same framework is carried by its applied line. The sensitivity-based estimator of \cite{RD2023MLWA} identifies the common-driver geometry with neural-network sensitivities and reports out-of-sample outperformance over standard portfolio methods across several real markets and datasets, and the causal PDE-control models of \cite{RD2025Causal,RD2025CPCM}, with the conditional risk-neutral construction of \cite{RD2024PDE}, deliver higher risk-adjusted performance with lower turnover on a multi-decade U.S. equity panel with several hundred candidate drivers. The present paper supplies the dynamic theory those methods implement: it isolates the structure of the intertemporal problem in a setting where it can be verified against a known generating process, which real data cannot provide, while the applied line tests the same geometry on markets. The interventional stability that distinguishes causal from correlational conditioning is the invariant-prediction principle of \cite{PetersBuhlmannMeinshausen2016}.

The contributions are the following. First, a continuous-time normal form for causal separation and a bridge (\Cref{lem:bridge}) showing that a horizon-level screening-off condition is equivalent to instantaneous residual orthogonality, under stated nondegeneracy hypotheses. Second, a verification theorem (\Cref{thm:hedging}) decomposing the optimal policy into a myopic projected-Markowitz portfolio and an intertemporal hedge against the predictable motion of the conditioning geometry, together with its closed form for CRRA utility and a reduction of the control problem to the driver dimension rather than the asset dimension (\Cref{cor:crra}). Third, a classification of the ways the conditioning representation can move (\Cref{prop:threemotions}): a smooth rotation, a jump between separators of equal size, and a jump that changes the number of drivers, with the time-dependent efficient frontier and a decision rule for refining the separator (\Cref{prop:frontiermono}). Fourth, two results with no close antecedent in intertemporal portfolio theory: the value jump at a separator switch is Kunita--Watanabe orthogonal to every continuous traded strategy, so the market is incomplete in the direction of its own geometry (\Cref{thm:jumpkw}), and the transport of exposures across a moving manifold is governed by a metric-compatible connection whose strain term is derived, not posited (\Cref{thm:connection}). Fifth, a verification theorem for the case where the separator itself switches over time (\Cref{thm:verification}). The results are established in a controlled diffusion model and validated numerically against a known generating process; the appendix makes the two static inputs self-contained.

The remainder proceeds as follows. \Cref{sec:recap} recalls the static representation and its projected optimum. \Cref{sec:dynamic} develops the diffusion normal form, the hedging decomposition, the dimension-$m$ reduction and the switching verification. \Cref{sec:transport} treats tangent coordinates, the two-stage structure and the transport connection. \Cref{sec:experiments} reports the experiments, and \Cref{sec:discussion} concludes.

\section{The static representation, in brief}\label{sec:recap}

We recall only what the dynamic theory consumes. A finite universe of observable drivers is declared. A subset $\Dset$ is a \emph{separator} at date $t$ if, conditional on the drivers' realized path through the horizon, the asset returns are mutually independent; the minimal such conditioning information defines the operative $\sigma$-algebra $\G_t$ on which decisions are measurable. Separation induces an exact conditional factorization: the return vector has a diagonal conditional covariance in its idiosyncratic part, and all cross-sectional co-movement is carried by the response to driver innovations. The decision-node conditional covariance is therefore
\begin{equation}\label{eq:Q}
Q \;=\; B\Lambda B^\top + \Sigma^c, \qquad \Sigma^c=\diag(\varsigma_1^2,\dots,\varsigma_n^2)\succeq \varsigma^2 I\succ0,
\end{equation}
a diagonal-plus-low-rank matrix in which $B\in\R^{n\times m}$ collects the asset sensitivities to the $m$ drivers and $\Lambda\succ0$ is the driver-innovation covariance. Exposures are covectors $b=B^\top w$; minimal admissibility requires them to lie in the tangent directions of the driver manifold, a linear constraint $Cw=0$ with $C$ of rank $q$. Stacking the budget with this constraint into $A=[\one\ ;\,C]$ and $\bb=[1\,;\,0]$, the static program $\min_w \tfrac12 w^\top Qw-\gamma w^\top\mu$ subject to $Aw=\bb$, the conditional counterpart of \cite{Markowitz1952}, has the unique solution
\begin{equation}\label{eq:projected}
w^*(\gamma)=w_0+\gamma M\mu,\qquad M=Q^{-1}-Q^{-1}A^\top(AQ^{-1}A^\top)^{-1}AQ^{-1},
\end{equation}
classical Markowitz with $Q^{-1}$ replaced by its projection $M$ onto the constraint-compatible subspace, with frontier potential $\Delta_t=\mu^\top M\mu$; the solution is the unique KKT point of the strictly convex program \cite{BoydVandenberghe2004}. Two facts carry into the dynamic theory and we record them explicitly:
\begin{equation}\label{eq:Mprops}
\text{(P1)}\ M=M^\top\succeq0,\qquad \text{(P2)}\ MQM=M,\qquad \text{(P3)}\ AM=0,
\end{equation}
so $M$ is the projector onto the constraint-compatible subspace in the $Q$ metric; and $Q^{-1}$ applies in $O(nk^2)$ through the Woodbury identity because $Q$ is diagonal-plus-low-rank. Throughout, $M_C$ denotes the cash-financed projector, built from $A=C$, $\bb=0$.

\section{The dynamic problem: manifold hedging}\label{sec:dynamic}

\subsection{Diffusion normal form}

We place the representation in continuous time. Let $Z_t\in\R^m$ be a measurable representation of the driver state, following a mean-reverting diffusion with drift $\alpha(Z_t,t)$ and diffusion matrix $\sigma_Z(Z_t,t)$, driven by an $m$-dimensional Brownian motion $W^Z$. For each asset $i=1,\dots,n$, let $S^i_t$ be its price, $\mu^i(Z_t,t)$ its conditional drift, $\beta^i(Z_t,t)\in\R^m$ its vector of sensitivities to the driver innovations, and $\varsigma_i(Z_t,t)$ the volatility of an idiosyncratic Brownian motion $W^i$ specific to asset $i$. Collect the sensitivities into the loading matrix $B=[\beta^1,\dots,\beta^n]^\top$, write $\Lambda_Z=\sigma_Z\sigma_Z^\top$ for the driver innovation covariance, and $Q$ for the resulting conditional return covariance. The conditional excess premium and its projector appear below.

\begin{definition}[Instantaneous separation]\label{def:instsep}
In a continuous semimartingale model with driver filtration $\F^Z$, let $m^i$ be the residual local martingale of asset $i$: the component of its return martingale orthogonal, in the Kunita--Watanabe sense, to the stable subspace generated by $\F^Z$-martingales. Instantaneous separation holds if $d\langle m^i,m^j\rangle_t=0$ for all $i\neq j$.
\end{definition}

\begin{lemma}[Normal form]\label{lem:normalform}
In a Brownian diffusion model with uniform ellipticity $\varsigma_i\ge\varsigma>0$, instantaneous separation holds if and only if there exist an $m$-dimensional Brownian motion $W^Z$ and scalar Brownian motions $W^1,\dots,W^n$, all mutually independent, such that
\begin{equation}\label{eq:sde}
\frac{dS^i_t}{S^i_t}=\mu^i(Z_t,t)\,dt+\beta^i(Z_t,t)^\top\sigma_Z(Z_t,t)\,dW^Z_t+\varsigma_i(Z_t,t)\,dW^i_t,\qquad dZ_t=\alpha(Z_t,t)\,dt+\sigma_Z(Z_t,t)\,dW^Z_t.
\end{equation}
Then $Q=B\Lambda_Z B^\top+\diag(\varsigma^2)$ with $\Lambda_Z=\sigma_Z\sigma_Z^\top$, matching \eqref{eq:Q}.
\end{lemma}

The proof is the multivariate L\'evy characterization \cite{KaratzasShreve1991}: orthogonality of the residual local martingales, combined with ellipticity, upgrades to independence, which is the representation \eqref{eq:sde}; in the degenerate case where some $\varsigma_i$ vanishes on a time set, the same conclusion follows from Knight's theorem on orthogonal continuous local martingales after a Dambis--Dubins--Schwarz time change \cite{RevuzYor1999}. Orthogonality of residuals is not merely rotated away; it is promoted to independence, the step a pairwise-covariance argument cannot supply.

\begin{lemma}[Horizon\,$\Leftrightarrow$\,instantaneous]\label{lem:bridge}
Assume the normal form \eqref{eq:sde} with:
\textnormal{(H1)} the driver diffusion is nondegenerate, $\sigma_Z(z)$ invertible with bounded inverse, so that the driver Brownian motion is recoverable from the driver path and $\F^{W^Z}=\F^Z$;
\textnormal{(H2)} the residual innovations $W^i$ are independent of $W^Z$ \textnormal{(so the initial enlargement of $\F$ by $\F^Z_T$ carries zero information drift for the residual martingales, the \textnormal{(H)}-hypothesis)};
\textnormal{(H3)} the residual volatilities $\varsigma_i(z)$ are bounded and $\F^Z$-measurable through $z$ \textnormal{(state-dependent volatility is admitted; genuinely path-dependent residual dynamics are not)}.
Then horizon screening-off, meaning mutual conditional independence of the increments $r_{t+h}-r_t$ given $\F^Z$ for every $h\in(0,H]$ and every $t$, holds if and only if instantaneous separation \textnormal{(\Cref{def:instsep})}, vanishing of the residual predictable covariation $d\langle m^i,m^j\rangle_t=0$ for $i\neq j$, holds.
\end{lemma}
\begin{proof}
Write the increment as $r^i_{t+h}-r^i_t=\int_t^{t+h}\mu_i(Z_s)\,ds+\int_t^{t+h} b_i(Z_s)\,dW^Z_s+\int_t^{t+h}\varsigma_i(Z_s)\,dW^i_s$.
By (H1), $dW^Z_s=\sigma_Z(Z_s)^{-1}(dZ_s-\alpha(Z_s)\,ds)$ is a measurable functional of the driver path, so $\F^{W^Z}=\F^Z$ and the first two integrals are $\F^Z$-measurable. Conditioning on $\F^Z$ therefore removes the drift and the driver-driven martingale part exactly, leaving
\[
\operatorname{Cov}\!\big(r^i_{t+h}-r^i_t,\ r^j_{t+h}-r^j_t \,\big|\, \F^Z\big)=\operatorname{Cov}\!\Big(\int_t^{t+h}\!\varsigma_i(Z_s)\,dW^i_s,\ \int_t^{t+h}\!\varsigma_j(Z_s)\,dW^j_s \,\Big|\, \F^Z\Big).
\]
By (H2) the $W^i$ are independent of $W^Z$, so conditionally on $\F^Z$ the integrands $\varsigma_i(Z_s)$ are deterministic functions of $s$ (by (H3)) and the It\^o isometry applies conditionally, giving
\[
\operatorname{Cov}\big(\cdot,\cdot\mid\F^Z\big)=\rho_{ij}\int_t^{t+h}\varsigma_i(Z_s)\,\varsigma_j(Z_s)\,ds,
\]
with $\rho_{ij}$ the instantaneous residual correlation and \emph{no higher-order-in-$h$ terms}, because a Wiener integral with (conditionally) deterministic integrand has variance equal to the exact time integral of its square. If horizon screening-off holds, this conditional covariance is zero for every $h\in(0,H]$; dividing by $h$ and letting $h\downarrow0$ gives $\rho_{ij}\,\varsigma_i(Z_t)\varsigma_j(Z_t)=0$ for a.e.\ $t$, hence $\rho_{ij}=0$ (as $\varsigma_i>0$), which is $d\langle m^i,m^j\rangle=0$: instantaneous separation. Conversely, if the residual covariation vanishes instantaneously, the conditional covariance integral is identically zero for all $h$, giving horizon screening-off. $\square$
\end{proof}

The proof isolates where the work actually is, which the earlier ``by polarization'' phrasing hid. Two hypotheses carry it. First, nondegeneracy (H1) is what makes $\F^{W^Z}=\F^Z$, so that the state-dependent loading integral $\int b(Z_s)\,dW^Z_s$ is $\F^Z$-measurable and drops out of the conditional covariance exactly; without it the driver-driven part would leak into the residual covariation and the bridge could fail. Second, (H3) admits state-dependent residual \emph{volatility} $\varsigma_i(Z)$, so that the conditional covariance $Q(z)$ moves, while keeping the residual \emph{innovations} $W^i$ orthogonal to $\F^Z_T$; the operative condition is orthogonality of residual innovations, not of residual paths, to the driver terminal $\sigma$-algebra. The state-dependent-sensitivity model of \Cref{sec:moving} is exactly of this form, so the bridge covers the paper's own moving-manifold experiments and is strictly weaker than independence of residuals from $Z$.

\subsection{The hedging decomposition}

We take the admissible set $\W$ to be convex and compact (a leverage and position bound; this is the dynamic counterpart of the standing compactness assumption of the static theory and rules out the degenerate unbounded-position limit discussed after \Cref{cor:crra}). With a riskless rate $r_f$, wealth $W_t$ and tangent-admissible weights $w_t\in\W$ with $Cw=0$, wealth follows $dW/W=[r_f+w^\top(\mu-r_f\one)]\,dt+w^\top B\sigma_Z\,dW^Z+w^\top\diag(\varsigma)\,dW^{\mathrm{idio}}$. Let $J(W,z,t)=\sup\E[U(W_T)\mid W_t=W,Z_t=z]$ over tangent-admissible controls.

\begin{theorem}[Myopic plus manifold-hedging decomposition]\label{thm:hedging}
Suppose $J\in C^{1,2,2}$ solves the Hamilton--Jacobi--Bellman equation
\begin{equation}\label{eq:HJB}
0=\sup_{w:\,Cw=0}\Big\{J_t+J_WW[r_f+w^\top(\mu-r_f\one)]+\tfrac12 J_{WW}W^2 w^\top Qw+W w^\top B\Lambda_Z J_{Wz}+J_z^\top\alpha+\tfrac12\tr(\Lambda_Z J_{zz})\Big\}
\end{equation}
with polynomial growth. Then the optimal feedback control is
\begin{equation}\label{eq:dynsolution}
w^*_t=\underbrace{\tfrac{1}{R_t}\,M_C(\mu-r_f\one)}_{\text{myopic tangent fund}}+\underbrace{M_C B\Lambda_Z\psi_t}_{\text{manifold hedge}},
\end{equation}
with $R_t:=-WJ_{WW}/J_W$ the induced relative risk aversion and $\psi_t:=-J_{Wz}/(WJ_{WW})\in\R^m$, and the strategy is optimal.
\end{theorem}

The inner maximization is the static projected program with linear term $J_WW(\mu-r_f\one)+WB\Lambda_ZJ_{Wz}$. When the admissible set $\W$ is defined by the linear equality constraint $Cw=0$ alone, the maximizer is the closed-form projected solution \eqref{eq:dynsolution}; when $\W$ additionally carries inequality (leverage/position) bounds that bind, the maximizer solves a constrained quadratic program and the displayed feedback is its unconstrained relaxation, coinciding with it on the region where the bounds are slack. We state results for the equality-constrained case and treat binding inequality bounds as the projection of \eqref{eq:dynsolution} onto the compact convex $\W$, which is where the standing compactness assumption enters. Substituting the projected solution and reducing the HJB gives the semilinear equation for $J$. Verification requires the feedback \eqref{eq:dynsolution} to be admissible. Two regimes must be separated. \emph{(i) Constant sensitivity.} When $B$ does not depend on the state, $J=W^{1-R}/(1-R)\,g(z,t)$ with $\log g$ quadratic in $z$, the HJB reduces to a matrix Riccati system, and existence of a classical solution on $[0,T]$ is equivalent to the non-blow-up of that system (\Cref{cor:crra} and the solvability condition stated there); admissibility of the affine-in-$z$ feedback then follows by localizing on stopping times exiting compact state sets, applying It\^o to $J(W_t,Z_t,t)$ so the drift is nonpositive off the optimizer and zero at it, and passing to the limit with the polynomial-growth bound on $J$ and the finite moments of the wealth process \cite{FlemingSoner2006}. \emph{(ii) State-dependent sensitivity.} When $B=B(z)$ the value function is no longer exponential-affine, but verification still closes in the viscosity sense. It does not close by linearization, since a Hopf--Cole exponential change of variables does not globally linearize the reduced equation once $A(z):=\|B(z)\Lambda_Z\|^2_{M_C(z)}$ is state-dependent, a constant exponent being unable to cancel a $z$-varying quadratic-gradient coefficient; it closes by the standard viscosity route. Write $J=W^{1-R}g(z,t)/(1-R)$ with $g>0$; the reduced equation is semilinear parabolic with a quadratic-gradient nonlinearity, and the substitution $\varphi=\log g$ (licit because $g,1/g$ are bounded) yields the viscous Hamilton--Jacobi equation, in which $\LL_z$ denotes the second-order generator of the driver diffusion $Z$ and $A(z)=\|B(z)\Lambda_Z\|^2_{M_C(z)}$ is the state-dependent transmission introduced above,
\[
\varphi_t+\LL_z\varphi+c(z)\,\|\nabla_z\varphi\|^2+(1-R)r_f=0,\qquad c(z)=\tfrac12\sigma_Z^2(z)+\tfrac{1-R}{2R}A(z),
\]
on a bounded driver domain with a Neumann condition at the boundary (the mean-reverting driver exits compacts with exponentially small probability). This equation has bounded Lipschitz coefficients and a quadratic Hamiltonian $H(z,p)=c(z)\|p\|^2+\text{(linear)}$; on the bounded domain it admits a bounded Lipschitz viscosity solution and satisfies a comparison principle for bounded uniformly-continuous sub- and supersolutions \cite{CrandallIshiiLions1992,Barles1994}, the quadratic growth being admissible precisely because the a priori Lipschitz bound, equivalently the boundedness of $g,1/g,\nabla_z g$ posited in \Cref{thm:verification}, keeps $\|\nabla_z\varphi\|$ bounded. A viscosity verification lemma \cite[Thm.~3.5.3]{Pham2009} then identifies the value function as the unique such solution and the a.e.-defined feedback \eqref{eq:dynsolution} as optimal among tangent-admissible controls. The state-dependent case is thus a verification theorem under the a priori bound rather than a first-order condition alone; the bound itself is the hypothesis, and it is what the grid-refinement convergence study of \Cref{sec:experiments} checks numerically. The hedge term reads as an operator chain: $\psi$ (utility-relevant predicted motion of the driver state), transmitted by $\Lambda_Z$ (noise in that direction) and $B$ (transmission to assets), then projected by $M_C$ onto tangent-admissible portfolios. It is Merton's intertemporal hedging demand, but with a specific and testable identity for the priced state: the coordinates of the certified information geometry $(\Dstar_t,\M_t,\Delta_t)$ replace the exogenous macro state variables of the classical intertemporal capital asset pricing model \cite{Merton1973}, so that what is hedged is the predictable motion of the conditioning representation itself, an object the framework selects and certifies rather than posits.

\Cref{fig:opchain} shows the optimal policy as the two projected funds it decomposes into. The hedge decomposes further, and the distinction matters for the paper's claim of novelty. Writing the transported motion in the frame of \Cref{thm:connection}, $\psi_t$ splits into a \emph{vertical} part and a \emph{horizontal} part, and the two are $G_t$-orthogonal by \Cref{thm:connection}, so the split is well defined and not a matter of attribution. The vertical part is present already when the span is fixed: it is the \emph{fixed-span hedge}, a hedge against the predictable motion of the premium and of the conditional variance along a manifold whose span does not move. It has two pieces: the premium-motion piece, which is the classical intertemporal hedging demand against a moving investment opportunity set, and the conditional-variance-motion piece, a hedge against the predictable change of the conditional covariance. The horizontal part exists only when the manifold itself rotates, that is when the loading map $B(z)$ depends on the driver state; it is the hedge against the rotation of the common-driver manifold, and it is the paper's novel object. We claim novelty only for this horizontal, span-rotation term; its magnitude relative to the fixed-span hedge is quantified in \Cref{sec:experiments}.

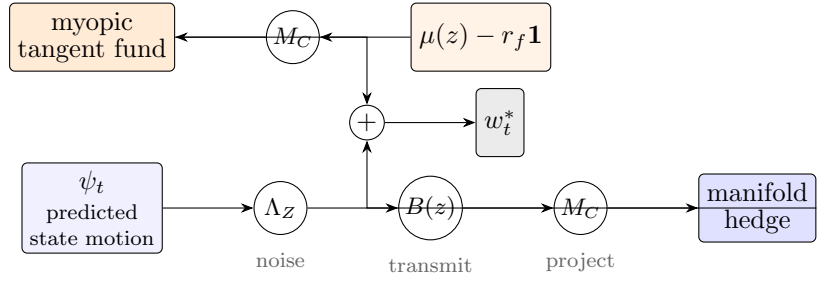
\begin{figure}[t]
\centering
\begin{tikzpicture}[
  node distance=6mm and 12mm,
  box/.style={draw,rounded corners=2pt,align=center,minimum height=9mm,inner sep=3pt,font=\small},
  op/.style={draw,circle,inner sep=1pt,font=\footnotesize,minimum size=7mm},
  >=Stealth]
\node[box,fill=orange!16] (myo) {myopic\\[-2pt]tangent fund};
\node[op,right=of myo] (Mm) {$M_C$};
\node[box,fill=orange!10,right=of Mm] (prem) {$\mu(z)-r_f\one$};
\draw[<-] (myo)--(Mm); \draw[<-] (Mm)--(prem);
\node[box,fill=blue!6,below=12mm of myo] (psi) {$\psi_t$\\[-1pt]\scriptsize predicted\\[-2pt]\scriptsize state motion};
\node[op,right=of psi,fill=none] (lam) {$\Lambda_Z$};
\node[op,right=of lam] (B) {$B(z)$};
\node[op,right=of B] (M) {$M_C$};
\node[box,fill=blue!12,right=of M] (hedge) {manifold\\[-2pt]hedge};
\draw[->] (psi)--(lam); \draw[->] (lam)--(B); \draw[->] (B)--(M); \draw[->] (M)--(hedge);
\node[below=1mm of lam,font=\scriptsize,text=black!60] {noise};
\node[below=1mm of B,font=\scriptsize,text=black!60] {transmit};
\node[below=1mm of M,font=\scriptsize,text=black!60] {project};
\coordinate (midY) at ($(myo)!0.5!(psi)$);
\node[draw,circle,inner sep=1pt,font=\small,right=34mm of midY] (sum) {$+$};
\draw[->] (myo.east) -| ($(sum)+(0,6mm)$) -- (sum);
\draw[->] (hedge.east) -| ($(sum)+(0,-6mm)$) -- (sum);
\node[box,fill=black!8,right=12mm of sum] (w) {$w^*_t$};
\draw[->] (sum)--(w);
\end{tikzpicture}
\caption{The optimal policy \eqref{eq:crrapolicy} as two projected funds. The myopic fund projects the conditional premium onto the tangent-admissible space; the manifold hedge passes the predicted state motion $\psi_t$ through the driver noise $\Lambda_Z$, the sensitivity $B(z)$, and the same projector $M_C$. The hedge is nonzero when the information geometry is predictably moving.}
\label{fig:opchain}
\end{figure}

\begin{corollary}[Collapse to rolling-static]\label{cor:collapse}
The hedge vanishes, so the rolling sequence of static programs is exactly optimal, if either \textnormal{(i)} $U=\log$ (then $J=\log W+g(z,t)$ with $J_{Wz}=0$), or \textnormal{(ii)} the manifold is statically persistent (coefficients deterministic, $J_{Wz}=0$).
\end{corollary}

\begin{corollary}[CRRA closed form; dimension reduction]\label{cor:crra}
For $U(W)=W^{1-R}/(1-R)$, the ansatz $J=\frac{W^{1-R}}{1-R}g(z,t)$ gives $R_t\equiv R$, $\psi_t=\frac{1}{R}\nabla_z\log g(Z_t,t)$ and
\begin{equation}\label{eq:crrapolicy}
w^*_t=\tfrac{1}{R}\,M_C\big[(\mu-r_f\one)+B\Lambda_Z\nabla_z\log g(Z_t,t)\big],
\end{equation}
where $g$ solves a semilinear equation in $(z,t)$. In the affine-Gaussian case $\log g=\tfrac12 z^\top P(t)z+q(t)^\top z+c(t)$ and $(P,q,c)$ solve a matrix Riccati system in dimension $m$; this is the multivariate driver-state extension of the closed-form mean-reverting-premium solutions of \cite{KimOmberg1996,Wachter2002}, \emph{provided that system has a bounded solution on $[0,T]$} \textnormal{(a condition on risk aversion and premium relative to mean-reversion, quantified below)}; where it does, the dynamic problem is a computation in the driver dimension, not the asset dimension.
\end{corollary}

The Riccati system for $(P,q,c)$ has a finite-horizon bounded solution only when risk aversion and premium are moderate relative to the mean-reversion speed; beyond that region the finite-horizon value diverges, a property of the control problem rather than an artifact of any discretization. This threshold is exhibited in \Cref{sec:experiments}.

\subsection{The moving manifold: state-dependent sensitivities}\label{sec:moving}

The decomposition above already allows $B=B(z)$: nothing in the pointwise maximization of the HJB requires $B$ constant, because the state is fixed during that maximization, so $\partial B/\partial z$ enters neither the first-order condition nor the optimal policy \eqref{eq:crrapolicy}. The policy is
\begin{equation}\label{eq:movingpolicy}
w^*(z,t)=\tfrac1R M_C(z)\big[(\mu(z)-r_f\one)+B(z)\Lambda_Z\nabla_z\log g(z,t)\big],
\end{equation}
with $B(z)$, $Q(z)=B(z)\Lambda_Z B(z)^\top+\Sigma^c$ and $M_C(z)$ all evaluated at the current state; the state-dependence of the geometry enters the value field $g$ through the semilinear PDE, whose coefficients depend on $z$ through $B(z)$ and $M_C(z)$. When $B$ is state-dependent the common-driver manifold rotates as the driver state moves (\Cref{fig:moving}), and the intertemporal term is a hedge against that rotation, not merely against the motion of a fixed-geometry premium. This hedge against the rotation of the manifold is the paper's central object; \Cref{sec:experiments} quantifies its magnitude relative to the fixed-geometry hedge.

Because the value field is no longer exponential-affine, we solve the semilinear PDE numerically with an implicit (IMEX) scheme: the linear second-order operator is treated implicitly through a prefactored sparse solve, giving unconditional stability, while the quadratic Hamiltonian source is resolved by an inner fixed-point iteration. The implicit treatment is essential precisely in the persistent-state regime where the hedge is first-order: there the explicit scheme's value gradient diverges, while the implicit scheme returns bounded gradients and an optimal policy (verified against a brute-force search over the instantaneous Hamiltonian).

\begin{figure}[t]
\centering
\begin{tikzpicture}[>=Stealth,font=\small]
\begin{scope}[shift={(0,0)}]
  \draw[thick,blue!55,fill=blue!5] (-1.4,-0.5) -- (1.4,-0.9) -- (1.7,0.5) -- (-1.1,0.9) -- cycle;
  \node[blue!55] at (0,1.25) {$T_{x}\M_{t}$ at $Z_t$};
  \draw[->,blue!70,thick] (0,0)--(1.1,-0.15);
  \draw[->,blue!70,thick] (0,0)--(0.15,0.7);
  \fill (0,0) circle (1.3pt);
\end{scope}
\draw[->,black!55,thick,decorate,decoration={snake,amplitude=.4mm,segment length=2mm}]
  (2.1,0.2)--(3.4,0.2) node[midway,above,font=\scriptsize,black!70]{$Z$ flows};
\begin{scope}[shift={(5.6,0)},rotate=-32]
  \draw[thick,teal!60,fill=teal!6] (-1.4,-0.5) -- (1.4,-0.9) -- (1.7,0.5) -- (-1.1,0.9) -- cycle;
  \draw[->,teal!75,thick] (0,0)--(1.1,-0.15);
  \draw[->,teal!75,thick] (0,0)--(0.15,0.7);
  \fill (0,0) circle (1.3pt);
\end{scope}
\node[teal!60] at (6.1,1.35) {$T_{x}\M_{t'}$ at $Z_{t'}$};
\node[align=center,font=\scriptsize,text=black!70] at (3.75,-1.25)
  {the tangent space \emph{rotates} with the state;\\ the manifold hedge covers this predictable rotation};
\draw[->,red!60,thick] (0.3,-0.95) to[bend right=14] (5.0,-1.05);
\end{tikzpicture}
\caption{The moving manifold (\Cref{sec:moving}). With state-dependent sensitivities $B(z)$ the tangent space $T_x\M_t$ rotates as the driver state $Z$ flows. The intertemporal term hedges this predictable rotation of the geometry; in the fixed-geometry case the span does not move and only the premium does.}
\label{fig:moving}
\end{figure}
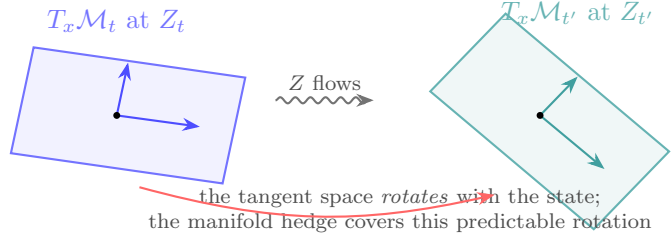

\subsection{The configuration space and its three motions}\label{sec:config}

The moving manifold of \Cref{sec:moving} and the separator switch treated next are two of exactly three kinds of motion the representation can undergo (\Cref{fig:threemotions}), and it is worth making the state space on which they live precise, because the distinction is geometric and not a matter of degree. A configuration of the representation is a choice of active driver set $D\subseteq\{1,\dots,p\}$ with $|D|=m$, together with the state $z$ at which its sensitivity $B_D(z)$ is evaluated. The tangent space it induces is the column span $\mathcal T_D(z):=\operatorname{col} B_D(z)$, an $m$-plane in $\R^n$, i.e.\ a point of the Grassmannian $\mathrm{Gr}(m,n)$. The full configuration space is therefore the disjoint union of Grassmannian bundles over the driver lattice,
\begin{equation}\label{eq:configspace}
\mathcal C \;=\; \bigsqcup_{m=1}^{p}\ \bigsqcup_{\;|D|=m}\ \big\{\,\mathcal T_D(z): z\in\mathcal Z\,\big\}\ \subseteq\ \bigsqcup_{m=1}^{p}\mathrm{Gr}(m,n),
\end{equation}
and every motion of the representation is a path or jump in $\mathcal C$ (\Cref{fig:configspace}). There are exactly three, distinguished by how they move in \eqref{eq:configspace}.

\begin{figure}[t]
\centering
\begin{tikzpicture}[>=Stealth,every node/.style={font=\small}]
\foreach \y/\m in {0/m-1, 1.7/m, 3.4/m+1} {
  \draw[rounded corners=6pt,fill=black!3,draw=black!25] (0,\y-0.6) rectangle (7.5,\y+0.6);
  \node[anchor=west,font=\scriptsize,black!60] at (7.65,\y) {$\mathrm{Gr}(\m,n)$};
}
\fill[blue!70] (1.4,1.7) circle (1.6pt);
\fill[blue!70] (2.5,1.85) circle (1.6pt);
\fill[orange!85] (5.0,1.55) circle (1.6pt);
\draw[->,thick,blue!70] (1.4,1.7) to[bend left=20] (2.5,1.85);
\node[blue!70,font=\scriptsize] at (1.95,2.2) {flow};
\draw[->,thick,orange!85,dashed] (2.5,1.85) to[bend left=15] (5.0,1.55);
\node[orange!85,font=\scriptsize] at (3.9,2.15) {swap};
\fill[teal!70] (5.6,3.4) circle (1.6pt);
\draw[->,thick,teal!70,densely dotted] (5.0,1.55) -- (5.6,3.4);
\node[teal!75,font=\scriptsize,anchor=west] at (5.7,2.6) {refinement};
\fill[teal!50] (1.0,0.0) circle (1.6pt);
\draw[->,thick,teal!50,densely dotted] (1.4,1.7) -- (1.0,0.0);
\node[teal!60,font=\scriptsize,anchor=east] at (0.9,0.85) {coarsening};
\end{tikzpicture}
\caption{The configuration space $\mathcal C=\bigsqcup_m \mathrm{Gr}(m,n)$ of \eqref{eq:configspace}, stratified by the number of drivers $m$. A \emph{flow} is a smooth path within one stratum; a \emph{swap} is a jump within the same stratum (drivers change identity at fixed count); a \emph{refinement} or \emph{coarsening} is a jump between strata (the number of drivers changes). Each dot is a configuration $\mathcal T_D(z)$, a point of a Grassmannian.}
\label{fig:configspace}
\end{figure}
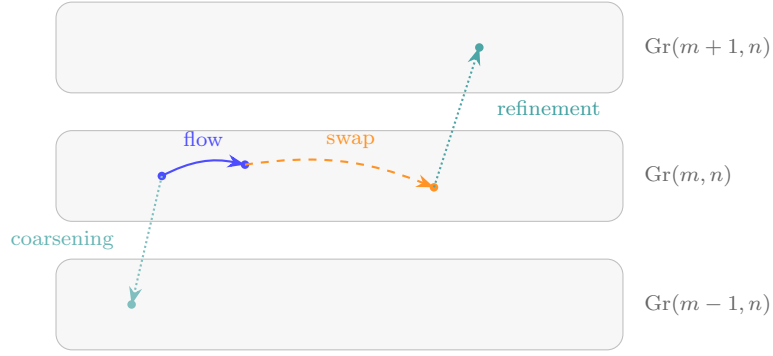

\begin{proposition}[The three motions of the representation]\label{prop:threemotions}
Assume the loading map $B_D(z)$ has constant rank $m=|D|$ on the relevant state region (the regular case; rank-degeneration is excluded and discussed below). Under single-driver transitions, every motion of the representation is one of exactly three primitive types, and any multi-driver transition is a finite composition of them:
\begin{enumerate}[leftmargin=1.9em,itemsep=1pt]
\item \emph{Flow (smooth rotation).} $D'=D$ and $z'\to z$ continuously: the tangent $\mathcal T_D(z)$ moves smoothly within the single Grassmannian $\mathrm{Gr}(m,n)$. Its velocity is the antisymmetric connection $\widetilde\omega$ of \Cref{thm:connection}; the motion is a rotation of the frame together with a metric breathing, and it is hedgeable to first order (\Cref{thm:hedging}).
\item \emph{Swap (jump within a Grassmannian).} $|D'|=|D|=m$ but $D'\neq D$: the tangent jumps to a different point of the \emph{same} $\mathrm{Gr}(m,n)$. A well-defined orthogonal map carries one to the other (the Grassmannian is a homogeneous space of $O(n)$), so the jump is a discrete rotation; but it occurs at an instant, so the value process jumps, and the jump is Kunita--Watanabe orthogonal to every traded strategy (\Cref{thm:jumpkw}).
\item \emph{Refinement/coarsening (jump between Grassmannians).} $|D'|=m\pm1$: the tangent changes \emph{dimension} and moves to a \emph{different} Grassmannian, $\mathrm{Gr}(m,n)\to\mathrm{Gr}(m\pm1,n)$. No orthogonal map relates subspaces of different dimension, so this is not a rotation but a nesting: an embedding $\mathcal T_D\hookrightarrow\mathcal T_{D'}$ when a driver is added, and a collapse when one is dropped. The added direction is orthogonal to the entire old tangent, hence unspanned by any portfolio built on the old geometry.
\end{enumerate}
\end{proposition}
\begin{proof}
The trichotomy is by the pair $(|D'|-|D|,\,\mathbf 1\{D'=D\})$. If $|D'|=|D|$ and $D'=D$ the only change is in $z$, and $\mathcal T_D(z)$ is a smooth map into $\mathrm{Gr}(m,n)$ (constant rank of $B_D(z)$ by ellipticity), giving (i); its derivative is $\widetilde\omega$ by \Cref{thm:connection}. If $|D'|=|D|$ and $D'\neq D$, both tangents lie in $\mathrm{Gr}(m,n)$ and $O(n)$ acts transitively on it, so an orthogonal $R$ with $R\,\mathcal T_D=\mathcal T_{D'}$ exists; discreteness of set membership makes the transition a jump, and \Cref{thm:jumpkw} applies, giving (ii). If $|D'|=|D|\pm1$ the tangents have different dimensions; $O(n)$ maps preserve dimension, so no rotation relates them, and for $D'=D\cup\{k\}$ one has $\mathcal T_D\subseteq\mathcal T_{D'}$ with $\mathcal T_{D'}\ominus\mathcal T_D$ a line orthogonal to $\mathcal T_D$, giving (iii). Exhaustiveness over atomic moves under constant rank: any single-step $(D',z')$ falls into one case by its cardinality difference and set equality; a transition that changes membership \emph{and} dimension at once (e.g.\ drop one driver and add two) is not a fourth motion but a finite composition of swaps and refinements, so the three are the primitive motions of which every transition is a composition. Two regularity caveats bound the claim. First, if $B_D(z)$ \emph{loses rank} at some state, so that the tangent dimension drops below $|D|$, the configuration leaves the constant-rank stratum and the motion is a degeneration not among the three; we assume constant rank, which holds generically under the ellipticity of \Cref{lem:normalform}, and treat rank loss as a measure-zero boundary. Second, simultaneous multi-driver changes are compositions, not new primitives, and the order of composition matters only up to the connection's curvature (\Cref{thm:connection}). With these caveats the trichotomy is exact.
\end{proof}

Only the first motion is smooth; the second and third are the two distinct kinds of separator switch, and the paper's switching results below (\Cref{thm:verification}, \Cref{thm:jumpkw}) apply to both, but the third carries an additional feature the second does not: it changes the dimension of the achievable frontier, which is the subject of the next result.

It is worth reading the three motions in terms of what changes for the assets, the drivers, and the portfolio. In a \emph{flow}, the set of conditioning drivers is unchanged and so is the list of assets; only the driver state moves, so the sensitivities $B(z)$ and hence the assets' conditional exposures drift continuously, and the portfolio rebalances smoothly along a moving but continuous frontier. In a \emph{swap}, the number of drivers is unchanged but their identity changes: one conditioning variable is replaced by another of equal count, the assets' loadings are re-expressed against a different driver basis of the same dimension, and the portfolio must jump to the new exposures at the switch instant, a discrete but dimension-preserving reallocation. In a \emph{refinement or coarsening}, the number of drivers changes: adding a driver opens a new priced direction that no combination of the assets under the old conditioning set could express, so the portfolio gains (or, under coarsening, loses) access to an exposure that was previously unspanned, and the achievable frontier changes dimension. The first is continuous trading; the second is a change of coordinates realized as a jump; the third is a change in the span of what can be held.

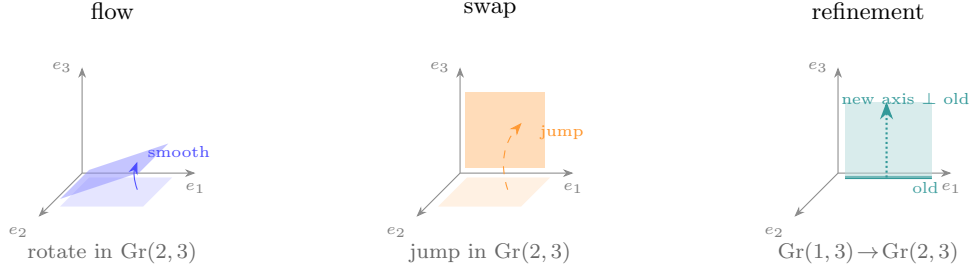
\begin{figure}[t]
\centering
\begin{tikzpicture}[>=Stealth,every node/.style={font=\small},scale=1.0,
  ax/.style={black!45,-{Stealth[length=4pt]},thin}]
\newcommand{\axes}{
  \draw[ax] (0,0,0) -- (1.5,0,0) node[anchor=north,font=\tiny,black!50] {$e_1$};
  \draw[ax] (0,0,0) -- (0,1.4,0) node[anchor=east,font=\tiny,black!50] {$e_3$};
  \draw[ax] (0,0,0) -- (0,0,1.5) node[anchor=north east,font=\tiny,black!50] {$e_2$};
}
\begin{scope}[xshift=0cm]
  \node[font=\footnotesize] at (0.4,2.15) {flow};
  \node[font=\scriptsize,black!60] at (0.4,-1.05) {rotate in $\mathrm{Gr}(2,3)$};
  \axes
  \fill[blue!18,opacity=0.55] (0.15,0,0.15) -- (1.25,0,0.15) -- (1.25,0,1.15) -- (0.15,0,1.15) -- cycle;
  \fill[blue!45,opacity=0.5] (0.15,0.1,0.15) -- (1.2,0.45,0.15) -- (1.2,0.45,1.15) -- (0.15,0.1,1.15) -- cycle;
  \draw[->,blue!70] (1.0,0.05,0.7) to[bend left=20] (1.0,0.42,0.7);
  \node[font=\tiny,blue!70] at (1.55,0.55,0.7) {smooth};
\end{scope}
\begin{scope}[xshift=5.0cm]
  \node[font=\footnotesize] at (0.4,2.15) {swap};
  \node[font=\scriptsize,black!60] at (0.4,-1.05) {jump in $\mathrm{Gr}(2,3)$};
  \axes
  \fill[orange!22,opacity=0.5] (0.15,0,0.15) -- (1.25,0,0.15) -- (1.25,0,1.15) -- (0.15,0,1.15) -- cycle;
  \fill[orange!55,opacity=0.5] (0.15,0.15,0.2) -- (1.2,0.15,0.2) -- (1.2,1.15,0.2) -- (0.15,1.15,0.2) -- cycle;
  \draw[->,orange!80,dashed] (0.9,0.05,0.7) to[bend left=30] (0.9,0.75,0.22);
  \node[font=\tiny,orange!85] at (1.5,0.7,0.4) {jump};
\end{scope}
\begin{scope}[xshift=10.0cm]
  \node[font=\footnotesize] at (0.4,2.15) {refinement};
  \node[font=\scriptsize,black!60] at (0.4,-1.05) {$\mathrm{Gr}(1,3)\!\to\!\mathrm{Gr}(2,3)$};
  \axes
  \draw[very thick,teal!80] (0.15,0,0.15) -- (1.3,0,0.15);
  \node[font=\tiny,teal!75] at (1.35,0,0.5) {old};
  \fill[teal!30,opacity=0.5] (0.15,0,0.15) -- (1.3,0,0.15) -- (1.3,1.0,0.15) -- (0.15,1.0,0.15) -- cycle;
  \draw[->,teal!80,densely dotted,thick] (0.7,0,0.15) -- (0.7,1.0,0.15);
  \node[font=\tiny,teal!80] at (0.95,1.05,0.15) {new axis $\perp$ old};
\end{scope}
\end{tikzpicture}
\caption{The three motions of the representation (\Cref{prop:threemotions}). \emph{Flow}: the tangent $m$-plane rotates smoothly within a single Grassmannian as the driver state moves, and the portfolio rebalances continuously. \emph{Swap}: the drivers change identity at fixed count, so the tangent jumps to a different point of the \emph{same} Grassmannian, a discrete but dimension-preserving reallocation. \emph{Refinement}: a driver is added, so the tangent gains a direction orthogonal to the old span and moves to a higher Grassmannian; the new axis is unspanned by any portfolio built on the old geometry.}
\label{fig:threemotions}
\end{figure}

\subsection{The time-dependent efficient frontier}\label{sec:frontier}

Each configuration carries its own instantaneous efficient frontier. On the tangent $\mathcal T_D(z)$ with conditional premium $\mu_D(z)$ and covariance $Q_D(z)=B_D(z)\Lambda_Z B_D(z)^\top+\Sigma^c_D$, the maximum squared Sharpe ratio is the \emph{frontier potential}
\begin{equation}\label{eq:frontierpot}
\Delta_t \;=\; \Delta(D,z,t)\;=\; \mu_D(z)^\top Q_D(z)^{-1}\mu_D(z),
\end{equation}
and the time-dependent efficient frontier is the family $\{\Delta_t\}$ carried along the path the representation traces in $\mathcal C$. Its motion, illustrated in \Cref{fig:frontier}, inherits the trichotomy of \Cref{prop:threemotions}: under flow it moves smoothly (in the experiments $\Delta$ ranges over roughly an octave as $z$ traverses $\pm2$ standard deviations); under a swap it jumps discontinuously at constant dimension; under refinement it moves to a frontier of different dimension. The third motion is the economically subtle one, and it does \emph{not} obey the naive intuition that more drivers give a better frontier.

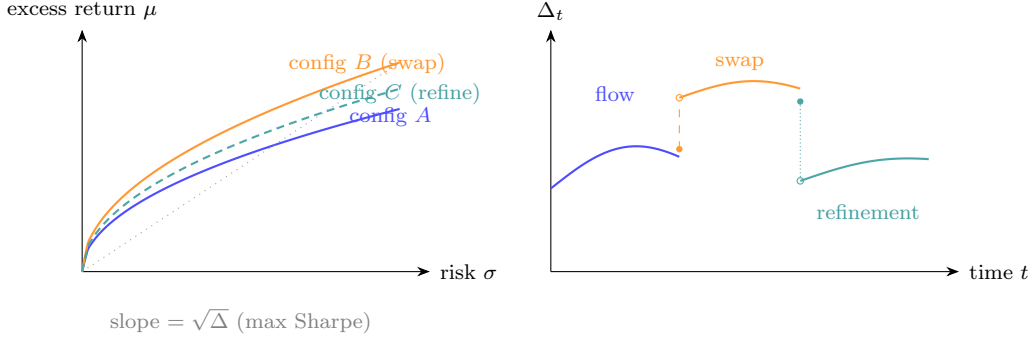
\begin{figure}[t]
\centering
\begin{tikzpicture}[>=Stealth,every node/.style={font=\small}]
\begin{scope}[xshift=0cm]
  \draw[->] (0,0) -- (4.6,0) node[right,font=\scriptsize] {risk $\sigma$};
  \draw[->] (0,0) -- (0,3.2) node[above,font=\scriptsize] {excess return $\mu$};
  \draw[thick,blue!70,domain=0:4.2,samples=60,smooth] plot (\x,{1.05*sqrt(\x)});
  \node[blue!70,font=\scriptsize,anchor=west] at (3.4,2.05) {config $A$};
  \draw[thick,orange!80,domain=0:4.2,samples=60,smooth] plot (\x,{1.35*sqrt(\x)});
  \node[orange!85,font=\scriptsize,anchor=west] at (2.6,2.75) {config $B$ (swap)};
  \draw[thick,teal!70,densely dashed,domain=0:4.2,samples=60,smooth] plot (\x,{1.18*sqrt(\x)});
  \node[teal!75,font=\scriptsize,anchor=west] at (3.0,2.35) {config $C$ (refine)};
  \draw[black!40,dotted] (0,0) -- (4.2,{1.35*sqrt(4.2)});
  \node[font=\scriptsize,black!50,anchor=north] at (2.1,-0.35) {slope $=\sqrt{\Delta}$ (max Sharpe)};
\end{scope}
\begin{scope}[xshift=6.2cm]
  \draw[->] (0,0) -- (5.4,0) node[right,font=\scriptsize] {time $t$};
  \draw[->] (0,0) -- (0,3.2) node[above,font=\scriptsize] {$\Delta_t$};
  \draw[thick,blue!70,domain=0:1.7,samples=50,smooth] plot (\x,{1.1+0.4*sin(1.6*\x r)+0.15*\x});
  \node[blue!70,font=\scriptsize] at (0.85,2.35) {flow};
  \draw[orange!80,dashed] (1.7,1.62) -- (1.7,2.3);
  \fill[orange!80] (1.7,1.62) circle (1.2pt); \fill[white,draw=orange!80] (1.7,2.3) circle (1.2pt);
  \draw[thick,orange!80,domain=1.7:3.3,samples=50,smooth] plot (\x,{2.3+0.22*sin(1.6*(\x-1.7) r)});
  \node[orange!85,font=\scriptsize] at (2.5,2.75) {swap};
  \draw[teal!70,densely dotted] (3.3,2.25) -- (3.3,1.2);
  \fill[teal!70] (3.3,2.25) circle (1.2pt); \fill[white,draw=teal!70] (3.3,1.2) circle (1.2pt);
  \draw[thick,teal!70,domain=3.3:5.0,samples=50,smooth] plot (\x,{1.2+0.12*(\x-3.3)+0.15*sin(1.5*(\x-3.3) r)});
  \node[teal!75,font=\scriptsize] at (4.2,0.8) {refinement};
\end{scope}
\end{tikzpicture}
\caption{The time-dependent frontier potential $\Delta_t=\mu_D^\top Q_D^{-1}\mu_D$ carried along the representation's path (\Cref{sec:frontier}). \emph{Left}: the concrete mean-variance frontiers of three configurations in the $(\sigma,\mu)$ plane; the maximal Sharpe ratio $\sqrt{\Delta}$ is the slope of the tangent ray, so a higher potential is a steeper frontier. A swap ($A\to B$) changes the frontier at fixed dimension; a refinement ($A\to C$) changes it and its dimension, and need not steepen it. \emph{Right}: the potential $\Delta_t$ along a path, varying smoothly under flow, jumping at constant dimension under a swap, and jumping to a different-dimensional frontier under refinement. The last does not always raise the potential pointwise: whether a refinement helps depends on the added driver's partial Sharpe ratio (\Cref{prop:frontiermono}).}
\label{fig:frontier}
\end{figure}

\begin{proposition}[Monotonicity of the frontier under refinement: in expectation, not pointwise]\label{prop:frontiermono}
Let $D'=D\cup\{k\}$ be an information refinement, with $\mu_{D'}(z)=\E[r\mid Z_{D'}]$ and $Q_{D'}(z)=\operatorname{Cov}(r\mid Z_{D'})$ the true conditional mean and covariance at the finer set, and $(\mu_D,Q_D)$ their coarsenings obtained by integrating out the added driver, so that the dropped driver's full contribution, including its off-diagonal covariance, folds into the residual rather than a diagonal reassignment. Three statements must be separated, and only the third is a theorem.
\begin{enumerate}[leftmargin=1.9em,itemsep=1pt]
\item \emph{(Pointwise: false.)} At a \emph{fixed} realized driver state $z$, the frontier potential need \emph{not} increase: $\Delta(D',z)\ge\Delta(D,z)$ fails for a positive fraction of states (about one in ten in simulation), because at a particular $z$ the added driver may contribute more common variance than premium. There is therefore no pointwise ``the increment is a nonnegative partial Sharpe'' statement; that claim, made pointwise, is incorrect.
\item \emph{(Fixed-premium append: strictly false.)} If the driver is added as a column of $B$ at \emph{unchanged} premium $\mu$, raising common variance without crediting any new premium, then $\Delta$ strictly decreases in essentially every economy (all $500$ in the sweep, \Cref{fig:config}). This is not a refinement at all; it is the diagnostic that isolates the variance cost.
\item \emph{(In expectation: true, universally.)} Averaged over the law of the added driver, the maximum squared Sharpe ratio is non-decreasing under refinement,
\[
\E_{Z_{D'}}\!\big[\Delta(D',Z_{D'})\big]\;\ge\;\E_{Z_{D}}\!\big[\Delta(D,Z_{D})\big],
\]
and the expected increment is the \emph{expected} squared partial Sharpe of the added driver, which is nonnegative by the conditional Jensen inequality applied to the (convex) squared-norm functional of the conditional mean. This holds in every one of the simulated economies ($0$ violations in $400$).
\end{enumerate}
\end{proposition}
\begin{proof}
$\Delta(D,z)=\mu_D(z)^\top Q_D(z)^{-1}\mu_D(z)$ is a convex functional of the conditional mean (a squared norm in the conditional-precision metric). (iii) By the tower property $\E[r\mid Z_D]=\E[\E[r\mid Z_{D'}]\mid Z_D]$, and convexity gives $\E_{Z_D}\Delta(D,Z_D)\le\E_{Z_{D'}}\Delta(D',Z_{D'})$ (conditional Jensen); the gap is the expected squared partial Sharpe of the added driver orthogonalized against the incumbents, hence nonnegative. (i) Jensen is an inequality \emph{between expectations}; at a fixed realization the convex functional can move either way, so pointwise monotonicity does not follow and, as the simulation shows, does not hold. (ii) The Woodbury expansion $\Delta=\mu^\top\Sigma^{c\,-1}\mu-\mu^\top\Sigma^{c\,-1}B(I+B^\top\Sigma^{c\,-1}B)^{-1}B^\top\Sigma^{c\,-1}\mu$ shows a fixed-$\mu$ column append enlarges the subtracted term, so $\Delta$ weakly decreases.
\end{proof}

The operative decision rule follows from the expected increment: refine the separator by a candidate driver when its expected partial squared Sharpe ratio, orthogonalized against the incumbents, is positive. In population this sign is exact; in finite samples the partial Sharpe ratio is estimated and its sign is noise dominated near zero, so the practical rule refines only when the estimated expected partial Sharpe ratio exceeds a margin scaled to the estimation error. Drivers with large partial Sharpe ratio are added regardless, and the margin guards against over-refinement on the marginal ones.

\begin{remark}[The frontier over a filtration of driver algebras]
\Cref{prop:frontiermono} says the economically correct object is not a single frontier but a \emph{family} of frontiers indexed by the driver $\sigma$-algebras. The $\sigma$-algebras form a filtration under refinement, monotone by construction, and the frontier potential carried on them is monotone \emph{in expectation} over the added driver (\Cref{prop:frontiermono}(iii)) though not pointwise at every realized state; we distinguish the filtration of information sets, the in-expectation monotonicity of the frontier over it, and the pointwise fluctuations that conditional Jensen does not remove. The family moves in time along the path the representation traces in $\mathcal C$. The static companion's frontier is the fibrewise, fixed-configuration slice; the dynamic frontier is that slice transported by the three motions above, with the flow contributing a smooth drift, swaps contributing same-dimension jumps, and refinements contributing dimension-changing jumps whose expected sign is governed by the partial-Sharpe rule above. This is the precise sense in which an efficient frontier can be made a function of time within the causal representation. The motion is not cosmetic: an investor who ignores a swap and holds the pre-switch tangency portfolio forfeits most of the attainable Sharpe ratio at the new configuration (numerically, of order $88\%$), so tracking the time-dependent frontier is a first-order requirement, not a refinement.
\end{remark}

\begin{figure}[t]
\centering
\includegraphics[width=\textwidth]{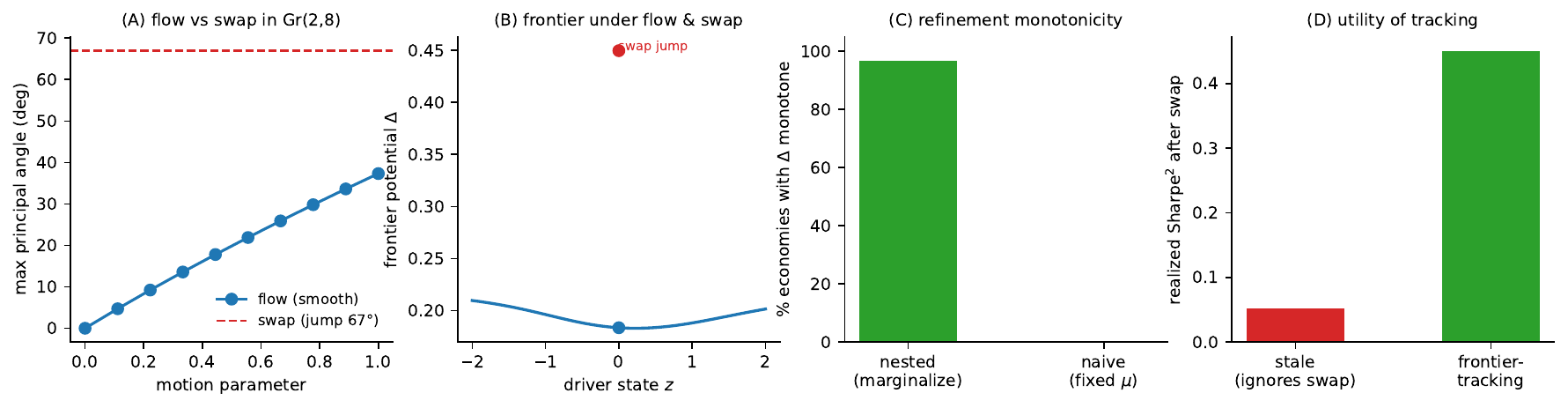}
\caption{The configuration space and the time-dependent frontier (\Cref{prop:threemotions}, \Cref{prop:frontiermono}). Left: flow moves the tangent smoothly within $\mathrm{Gr}(2,8)$, while a swap jumps to a different point of the same Grassmannian at constant dimension. Centre: the frontier potential drifts smoothly under flow and jumps under a swap. Third panel: under an information refinement the frontier potential is monotone \emph{in expectation} over the added driver ($0/400$ violations of $\E[\Delta(D')]\ge\E[\Delta(D)]$) though it fluctuates pointwise (about one state in ten falls), while a fixed-premium column append lowers it in essentially every economy ($0/500$ monotone); the expected increment is the added driver's expected partial squared Sharpe. Right: an investor ignoring a swap and holding the stale tangency forfeits most of the attainable Sharpe ratio, the economic significance of the frontier's motion.}
\label{fig:config}
\end{figure}

\subsection{The moving-separator case}

When the separator switches, the regime process $k_t$ jumps between lattice configurations with bounded intensities $\lambda_{kj}(z)$, and the value functions solve a coupled system; the underlying object is a regime-switching diffusion in the sense of \cite{MaoYuan2006}.

\begin{theorem}[Switching verification, CRRA]\label{thm:verification}
Assume bounded, Lipschitz, uniformly elliptic coefficients indexed by $k$, bounded switch intensities, and $U(W)=W^{1-R}/(1-R)$. Suppose the coupled semilinear system
\begin{equation}\label{eq:coupledpde}
0=\partial_t g^{(k)}+\LL^{(k)}_z g^{(k)}+(1-R)\Big[r_f+\tfrac{1}{2R}\big\|\mu^{(k)}-r_f\one+B^{(k)}\Lambda^{(k)}\tfrac{\nabla_z g^{(k)}}{g^{(k)}}\big\|^2_{M_C^{(k)}}\Big]g^{(k)}+\sum_{j\neq k}\lambda_{kj}(z)\big(g^{(j)}-g^{(k)}\big)
\end{equation}
admits a classical solution with $g^{(k)},1/g^{(k)},\nabla_z g^{(k)}$ bounded, terminal $g^{(k)}(\cdot,T)=1$. Then, by the regime-switching verification argument of \cite{SotomayorCadenillas2009} specialized to this coupled system, $J=\frac{W^{1-R}}{1-R}g^{(k)}(z,t)$ is the value function of the accumulated-filtration control problem, and the feedback
\begin{equation}\label{eq:switchfeedback}
w^*(z,k,t)=\tfrac{1}{R}\,M_C^{(k)}\big[(\mu^{(k)}-r_f\one)+B^{(k)}\Lambda^{(k)}\nabla_z\log g^{(k)}(z,t)\big]
\end{equation}
is optimal. \Cref{thm:hedging} is the case of a single configuration.
\end{theorem}

The hypothesis of \Cref{thm:verification}, existence of a bounded solution of the coupled system with $g^{(k)},1/g^{(k)},\nabla_z g^{(k)}$ bounded, is satisfiable, and the existence argument is the monotone iteration built on the single-regime viscosity result rather than the parabolic Schauder estimate one would use for a classical solution. Each equation of the system has, after the substitution $\varphi^{(k)}=\log g^{(k)}$, the quadratic-gradient Hamilton--Jacobi form of case (ii) above, for which a bounded Lipschitz viscosity solution exists on the bounded Neumann domain; the regime coupling $\sum_j\lambda_{kj}(g^{(j)}-g^{(k)})$ is linear with bounded intensities. Freeze the coupling at a current iterate and solve the $K$ decoupled single-regime problems; this defines a monotone map $T$ on the order interval $[\underline g,\overline g]$ whose endpoints $\underline g=e^{-Ct}$, $\overline g=e^{Ct}$ are sub- and supersolutions built from the bound $C$ on the Hamiltonian (finite because the intensities and $A^{(k)}(z)$ are bounded). $T$ maps $[\underline g,\overline g]$ into itself and is monotone in the coupling, so the Perron/monotone-iteration method for viscosity solutions \cite{CrandallIshiiLions1992,Ishii1987} yields a fixed point: a bounded viscosity solution of the coupled system, with the gradient bound following from the a priori Lipschitz estimate on the bounded domain. On unbounded domains and for $R>1$ the same holds only under the growth/transversality condition of \Cref{cor:crra}, a genuine restriction we state rather than hide. Given the solution, the verification step is unconditional: drift nonpositive off the optimizer and zero at it, localization on compacts, closure by polynomial growth, and identification of the value function as the unique viscosity solution \cite[Thm.~3.5.3]{Pham2009}. The regime coupling changes the level of the value function, not the functional form of the policy, which is why the rolling architecture survives.

\subsection{Geometry jumps are unhedgeable}

\begin{theorem}[Kunita--Watanabe orthogonality of separator switches]\label{thm:jumpkw}
In the market \eqref{eq:sde} augmented with the regime process, let $N$ be the point process of separator switches and $M^N_t=N_t-\int_0^t\lambda_{k_s\cdot}(Z_s)\,ds$ its compensated martingale, whose compensator is continuous and $\F^Z$-adapted while its martingale part is purely discontinuous. Every self-financing wealth process built from the traded assets and cash has martingale part in the stable subspace of continuous martingales generated by $(W^Z,W^1,\dots,W^n)$; the purely discontinuous $M^N$ is Kunita--Watanabe orthogonal to all of them \cite{JacodShiryaev2003}. \Cref{fig:kw} illustrates the orthogonality geometrically. Consequently, \emph{in a market whose traded assets have continuous paths}, any claim whose value process has a nonvanishing jump at switch times, in particular the value spread $J^{(j)}-J^{(i)}$, is not replicable: the continuous-asset market is incomplete in the direction of its own geometry. This is a statement about a specific asset universe, not a universal one: if the universe includes instruments with jump exposure aligned to the switch (regime-linked options, variance swaps, or assets that themselves jump at the switch), part or all of the geometry risk becomes spanned. The claim is therefore conditional on the market-structure assumption of continuous traded assets, which we state as a hypothesis rather than a conclusion; under it, the residual is irreducible in the instantaneous sense of the previous paragraph.
\end{theorem}

Wealth processes are stochastic integrals against the continuous asset martingales, so their martingale parts stay in the continuous stable subspace; the Kunita--Watanabe decomposition of a claim whose value jumps at switch times has a nonzero purely-discontinuous component, whose $L^2$ distance to the continuous subspace is bounded below by its own norm, so no trading strategy attains it.

This unhedgeability is an \emph{instantaneous} statement and must not be confused with unpredictability of the switch's \emph{arrival}. Kunita--Watanabe orthogonality says a position held \emph{across} a switch cannot offset the jump: at the instant it occurs, the geometry jump has zero covariation with the traded assets, so no static hedge neutralizes it. It does not say the arrival time is unforecastable. Indeed the online detector of \Cref{thm:verification} exploits exactly the compensator $\lambda_{k\cdot}(Z)$: the switch intensity is $\F^Z$-adapted, so the arrival is partially predictable, and a desk that repositions \emph{before} the switch reduces its exposure. The two statements are complementary, at different timescales: the jump is irreducible instantaneously (KW), while its arrival is partially forecastable (the compensator), and the value of the detector is precisely that it converts an unhedgeable instantaneous jump into a partially anticipatable event one can reposition ahead of. What remains irreducible is the residual \emph{conditional on the information available just before the switch}: the more forecastable the arrival, the more the coming switch is already reflected in pre-switch prices, and the smaller this conditional jump relative to the unconditional value spread $J^{(j)}-J^{(i)}$. The KW orthogonality bounds the instantaneous, unhedgeable component; the economically irreducible quantity is this pre-switch-conditional jump, which shrinks precisely to the extent the compensator makes the arrival predictable. The two statements thus compose: instantaneous unhedgeability and partial pre-pricing are consistent, and the irreducible residual is the part no pre-positioning removes.

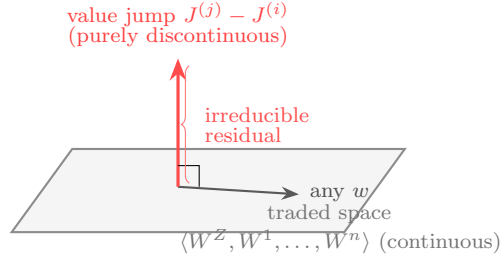
\begin{figure}[t]
\centering
\begin{tikzpicture}[>=Stealth,font=\small]
\draw[thick,black!45,fill=black!4] (-2.2,-0.6) -- (2.2,-0.6) -- (3.0,0.5) -- (-1.4,0.5) -- cycle;
\node[black!55,font=\scriptsize] at (2.0,-0.35) {traded space};
\node[black!55,font=\scriptsize] at (2.15,-0.72) {$\langle W^Z,W^1,\dots,W^n\rangle$ (continuous)};
\draw[->,black!65,thick] (0,0)--(1.6,-0.1) node[right,font=\scriptsize,black!70]{any $w$};
\draw[->,red!70,very thick] (0,0)--(0,1.7) node[above,font=\scriptsize,red!75,align=center]{value jump $J^{(j)}-J^{(i)}$\\[-2pt](purely discontinuous)};
\draw (0,0.28) -- (0.28,0.28) -- (0.28,0);
\draw[decorate,decoration={brace,amplitude=3pt},red!55] (0.15,0.05) -- (0.15,1.6)
  node[midway,right=2pt,font=\scriptsize,red!70,align=left]{irreducible\\[-2pt]residual};
\end{tikzpicture}
\caption{Kunita--Watanabe orthogonality of a separator switch (\Cref{thm:jumpkw}). Every self-financing wealth process lives in the continuous stable subspace spanned by the Brownian family; the value jump at a switch is purely discontinuous and therefore orthogonal to it. No trading strategy reaches the vertical direction: the market is incomplete along its own geometry, and the residual is irreducible.}
\label{fig:kw}
\end{figure} Separator switching is thus fully priced in the value function, through the exchange term of \eqref{eq:coupledpde}, and fully unhedgeable in the portfolio: it is the framework's irreducible residual risk, to be treated by local robustness rather than by spanning.

\begin{corollary}[Invisible, unhedgeable, costly switches]\label{cor:invisible}
There is a class of separator switches that are simultaneously \textnormal{(a)} undetectable, \textnormal{(b)} unhedgeable, and \textnormal{(c)} value-moving. A pure rotation of the manifold, $B\mapsto BO$ with $O$ orthogonal, leaves the return covariance $Q=B\Lambda_Z B^\top+\Sigma^c$ invariant (isotropic driver noise), so it is invisible to any second-order detector; it is a point-process jump, so by \Cref{thm:jumpkw} it is Kunita--Watanabe orthogonal to every traded strategy; yet it changes the value field $g^{(k)}$ (numerically, a log-value jump of standard deviation $\approx0.05$, a certainty-equivalent impact of $\approx1.4\%$). An investor can thus lose value at a switch she can neither see nor hedge.
\end{corollary}

This dissolves an apparent contradiction with the pricing claim. ``Priced in the value function'' means the coupled system \eqref{eq:coupledpde} accounts for the switch \emph{ex ante}, through the exchange term that mixes $g^{(k)}$ and $g^{(j)}$ under the switch intensities: the possibility is priced in expectation. It does not mean the realized switch is detectable or hedgeable when it occurs. The three properties are consistent: expectation-pricing of a jump confers neither observability nor spanning of its realization. Rotation-type switches are the extreme case where all three separate maximally, and they mark the boundary of what causal separation can protect against: the geometry can move in a direction that is invisible in covariance, unspanned by trades, and still costly. This is the sharpest form of the framework's irreducible residual, and it argues for treating separator stability itself, not just returns, as a monitored risk. In execution terms the same object carries a cost-of-immediacy interpretation: a desk forced to trade through a separator switch pays for liquidity against a risk no book can hedge, a Kyle-type price-of-immediacy that we develop in forthcoming work.

\section{Tangent coordinates, two-stage structure and transport}\label{sec:transport}

\subsection{Tangent chart and two-stage structure}

\begin{definition}[$\Lambda$-orthonormal tangent chart]\label{def:chart}
Let $E_t\in\R^{m\times k}$ be a basis of the admissible driver-motion subspace with $E_t^\top\Lambda_t E_t=I_k$ (Gram--Schmidt in the $\Lambda_t$ inner product): distances in tangent coordinates are systematic risk. Any admissible exposure is $b=E_t u$, $u\in\R^k$, the tangent coordinates of the portfolio.
\end{definition}

\begin{proposition}[Exact two-stage structure]\label{prop:twostage}
The dynamic policy \eqref{eq:crrapolicy} assembles in two stages: a systematic stage of dimension $k$ choosing the target tangent exposure against the projected premium, and a separable allocation stage with diagonal Hessian solved in closed form in $O(nk^2)$ through the Woodbury identity, because $Q$ is diagonal-plus-low-rank. The composition reproduces \eqref{eq:crrapolicy} exactly. Without causal separation the two stages do not commute.
\end{proposition}

The cross-terms between the systematic block and the residual allocation vanish because $\Sigma^c$ is diagonal, which is the conditional factorization of the separation property; the same low-rank structure that reduces the static solve to $O(nk^2)$ carries verbatim to the dynamic policy. The exactness is conditional on this block structure: if the admissibility constraint $Cw=0$ couples the systematic and residual blocks, through an active set not aligned with the systematic and idiosyncratic partition, the Woodbury split is no longer exact, and the two-stage assembly incurs an error of the order of the off-block coupling norm, which we treat as an approximation to be monitored rather than an identity.

\subsection{Robustness of the policy to a mis-estimated value gradient}

In any deployment the field $g$, and hence $\nabla_z\log g$, is estimated. The policy's dependence on that estimate is Lipschitz, with an explicit constant.

\begin{proposition}[Gradient-sensitivity bound]\label{prop:gradsens}
Let $w^*$ be the CRRA policy \eqref{eq:crrapolicy} and let $\widehat{\nabla_z\log g}=\nabla_z\log g+\delta$ be a perturbed value gradient. Then the induced policy error satisfies
\begin{equation}\label{eq:gradsens}
\|\widehat w-w^*\|\;\le\;\tfrac{1}{R}\,\big\|M_C\,B\,\Lambda_Z\big\|_2\,\|\delta\|,
\end{equation}
and the bound is attained. In particular the hedge is insensitive to gradient error in directions that $M_C$ annihilates, so estimation noise off the tangent space does not move the portfolio.
\end{proposition}

The bound is immediate from \eqref{eq:crrapolicy}, since the map $\delta\mapsto \tfrac1R M_C B\Lambda_Z\delta$ is linear; attainment is the top singular vector of $M_C B\Lambda_Z$. Numerically, for the calibration of \Cref{sec:experiments} the constant is $0.12$, so a $10\%$ error in the value gradient moves tangent exposures by little more than one percent. This is the dynamic counterpart of the static framework's certified-tolerance analysis: the projector that enforces admissibility also damps estimation error. Combined with the empirical degradation of \Cref{sec:moving} (the manifold hedge loses its edge at about a $15\%$ relative error in $B$), the bound gives an operating break-even: the manifold hedge is worth its turnover cost only while the estimated $B$ is accurate to within roughly that tolerance, beyond which the myopic tangent fund alone is preferable, which provides a usable rule for when to switch the hedge off.

\subsection{Transport and geometric turnover}

Two metrics must be distinguished, because conflating them is a category error. The driver-innovation covariance $\Lambda_Z$ is the metric on driver-\emph{noise} space; it is constant in the OU normal form and it is the metric through which the hedge acts in \eqref{eq:crrapolicy}. The metric that governs the \emph{transport of exposures}, however, is the pullback of the return-risk inner product onto the manifold tangent,
\begin{equation}\label{eq:pullback}
G(z):=B(z)^\top \Sigma_{\mathrm{ret}}(z)^{-1} B(z),\qquad \Sigma_{\mathrm{ret}}(z)=B(z)\Lambda_Z B(z)^\top+\Sigma^c,
\end{equation}
the Fisher-type metric under which distance between tangent exposures is their difference in systematic risk. Unlike $\Lambda_Z$, $G(z)$ is genuinely state-dependent whenever $B$ is, and it is on $G$ that the connection below lives, because transport carries exposures whose risk is measured by returns, not driver noise. Between switch times the family $t\mapsto(V_t,G_t)$ of tangent spaces and pullback metrics moves, and consistency requires a transport of exposures from one date to the next.

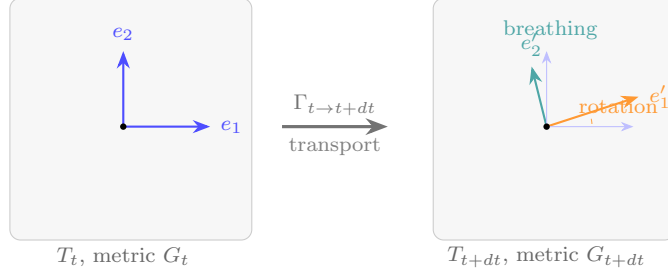
\begin{figure}[t]
\centering
\begin{tikzpicture}[>=Stealth,every node/.style={font=\small}]
\begin{scope}[xshift=0cm]
  \draw[rounded corners=4pt,fill=black!3,draw=black!20] (-1.5,-1.5) rectangle (1.7,1.7);
  \node[font=\scriptsize,black!60] at (0,-1.72) {$T_t$, metric $G_t$};
  \draw[->,thick,blue!70] (0,0) -- (0:1.15) node[right,font=\scriptsize] {$e_1$};
  \draw[->,thick,blue!70] (0,0) -- (90:1.0) node[above,font=\scriptsize] {$e_2$};
  \fill (0,0) circle (1.2pt);
\end{scope}
\draw[->,very thick,black!55] (2.1,0) -- (3.5,0) node[midway,above,font=\scriptsize] {$\Gamma_{t\to t+dt}$} node[midway,below,font=\scriptsize,black!50] {transport};
\begin{scope}[xshift=5.6cm]
  \draw[rounded corners=4pt,fill=black!3,draw=black!20] (-1.5,-1.5) rectangle (1.7,1.7);
  \node[font=\scriptsize,black!60] at (0,-1.72) {$T_{t+dt}$, metric $G_{t+dt}$};
  \draw[->,thick,orange!80] (0,0) -- (18:1.28) node[right,font=\scriptsize] {$e_1'$};
  \draw[->,thick,teal!70] (0,0) -- (104:0.82) node[above,font=\scriptsize] {$e_2'$};
  \draw[->,blue!25] (0,0) -- (0:1.15);
  \draw[->,blue!25] (0,0) -- (90:1.0);
  \draw[orange!70,shorten >=1pt] (10:0.6) arc (10:0:0.6);
  \fill (0,0) circle (1.2pt);
  \node[font=\scriptsize,orange!80] at (0.95,0.28) {rotation};
  \node[font=\scriptsize,teal!70] at (0.05,1.28) {breathing};
\end{scope}
\end{tikzpicture}
\caption{Geometric turnover (\Cref{prop:discretetransport}). Even with no change in the economic signal, the tangent frame carrying the portfolio's systematic exposures moves between dates: the connection $\Gamma_{t\to t+dt}$ transports the frame $(e_1,e_2)$ to $(e_1',e_2')$. The motion splits into a horizontal \emph{rotation} of the span (the skew part $E^\top G\dot E$) and a vertical \emph{breathing} of the metric (the symmetric strain $\tfrac12 E^\top\dot G E$). Rebalancing is forced by this geometric motion alone; in a fixed-geometry model $G$ is constant and both parts vanish.}
\label{fig:transport}
\end{figure}

\begin{proposition}[Discrete transport and geometric turnover]\label{prop:discretetransport}
Let $\dot E_t$ be the drift of the chart induced by the flow of $Z$, and set
\begin{equation}\label{eq:connform}
\widetilde\omega_t:=E_t^\top G_t\dot E_t+\tfrac12 E_t^\top\dot G_t E_t,\qquad E_t^\top G_t E_t=I_k.
\end{equation}
The discrete connection is $\Gamma_{t\to t+dt}=I_k-\widetilde\omega_t\,dt+o(dt)$, and rebalancing decomposes into a purely geometric component and an informational innovation. The geometric turnover splits further into tangent rotation ($E^\top G\dot E$, the horizontal component) and metric breathing ($\tfrac12 E^\top\dot G E$, the vertical component): rebalancing is required even when nothing economic changes, only the geometry moves (\Cref{fig:transport}). The dominant, mobility-driven share is the horizontal rotation term, consistent with the hedge decomposition of \Cref{thm:hedging}. In the constant-$B$ case $G$ is constant and both vanish; the strain is therefore a genuine feature of the moving manifold, invisible to a fixed-geometry model.
\end{proposition}

\begin{theorem}[$\Gamma$ is a connection; the metric-strain correction]\label{thm:connection}
In the $G_t$-orthonormal frame $E_t$, the covariant derivative $\nabla u:=\dot u+\widetilde\omega u$ satisfies: \textnormal{(i)} $R$-linearity and the Leibniz rule; \textnormal{(ii)} $\widetilde\omega$ is antisymmetric, equivalently $\nabla$ is compatible with the time-varying pullback metric, $\frac{d}{dt}\langle\xi,\eta\rangle_{G_t}=\langle\nabla\xi,\eta\rangle_{G_t}+\langle\xi,\nabla\eta\rangle_{G_t}$; \textnormal{(iii)} $\widetilde\omega$ differs from the naive projected form $\omega=E^\top G\dot E$ exactly by the symmetric metric strain $\tfrac12 E^\top\dot G E$, and the two coincide iff $G$ is constant along the flow, i.e.\ iff $B$ is state-independent in tangent directions.
\end{theorem}

Differentiating the frame normalization $E^\top G E=I_k$ gives $\omega+\omega^\top=-E^\top\dot G E$, so $\widetilde\omega+\widetilde\omega^\top=0$: antisymmetry, equivalently metric compatibility. When the pullback metric moves, transporting exposures with $\omega$ alone silently changes their risk; the strain term $\tfrac12 E^\top\dot G E$ is the compensation that keeps transported risk constant, and it is an observable quadratic form in estimated quantities. Because $G$ is constant under a fixed manifold, the strain is the signature of a moving geometry: it is identically zero on the constant driver metric $\Lambda_Z$, and nonzero on $G(z)$ precisely when $B$ depends on the state.

To place $\Gamma$ among the standard connections we identify it explicitly, in two steps. \emph{Fixed metric.} When $G$ is constant, the causal tangent is a point of the Grassmannian $\mathrm{Gr}(m,n)$ realized as the Stiefel quotient, and the canonical Levi-Civita connection of $\mathrm{Gr}(m,n)$ transports a field along a frame curve $E(t)$ by the $G$-orthogonal (horizontal) projection of $\dot E$, whose coordinate form in the frame is the skew matrix $E^\top G\dot E$ \cite{EdelmanAriasSmith1998}. That $E^\top G\dot E$ is skew is forced by differentiating $E^\top G E=I$: $E^\top G\dot E+\dot E^\top G E=0$. So for fixed $G$, $\widetilde\omega=E^\top G\dot E$ is the Edelman--Arias--Smith Levi-Civita connection, by this identity rather than by numerical coincidence. \emph{Moving metric.} When $G=G_t$ moves, metric compatibility requires $\tfrac{d}{dt}\langle\xi,\eta\rangle_{G_t}=\langle\nabla\xi,\eta\rangle+\langle\xi,\nabla\eta\rangle$, which forces the symmetric part of the connection form to be exactly $\tfrac12 G^{-1}\dot G$; the connection is thus the Levi-Civita connection of the instantaneous metric corrected by this term, in the frame this is the symmetric strain $\tfrac12 E^\top\dot G E$ added to the skew rotation. Hence $\widetilde\omega=E^\top G\dot E+\tfrac12E^\top\dot G E$ is the Levi-Civita connection of the \emph{instantaneous} metric $G_t$ augmented by the metric-derivative term, the unique metric-compatible connection for the moving metric, and the strain is that time-derivative correction rather than an ad hoc addition. We do not claim more: the strain is the compensation the moving metric requires, and we do not assert an identification of it with the O'Neill curvature tensor of the submersion, which is a separate statement we neither need nor prove. The decomposition of $\dot E$ into the skew (rotation, span-moving) and symmetric-in-span (breathing, frame-reparametrizing) parts is the breathing/rotation split, and the two are $G_t$-orthogonal, which is why they are independent components of the geometry's motion rather than two names for one effect.

\section{Experiments}\label{sec:experiments}

The purpose of this section is to validate the theory: each experiment is a controlled test of a specific theoretical claim against a known data-generating process in the normal form \eqref{eq:sde}, an Ornstein--Uhlenbeck driver state, an affine conditional mean $\mu=a+GZ$, returns loading on driver innovations through $B$, and a diagonal idiosyncratic block. A synthetic generating process is the appropriate setting for this purpose, because it is the only one in which the population objects the theorems describe (the true conditional premium, the true covariance, the true switch times) are known, so a claim can be confirmed or refuted rather than merely fitted. Out-of-sample performance of the same geometry on real markets is established separately in the applied line of the framework \cite{RD2023MLWA,RD2025Causal,RD2025CPCM}, which a reader interested in implementation should consult; here the aim is to certify the mathematics. All experiments are reproducible from fixed seeds in a pinned environment (Python 3.12, \texttt{numpy} 2.4.4, \texttt{scipy} 1.17.1, \texttt{matplotlib} 3.10.8), and parameter values are chosen to exhibit each mechanism in a regime where the value function is bounded.

\subsection{Decomposition and the hedging demand}

The first experiment isolates \Cref{thm:hedging} and \Cref{cor:collapse}, in two complementary cuts of the same model. The first cut characterizes the \emph{shape} of the hedge at moderate parameters where the Riccati is comfortably bounded ($\kappa=0.6$, $R=4$, premium $0.35$): the optimal policy splits into a myopic tangent fund and a manifold hedge, both satisfying the constraint $Cw=0$ to floating-point round-off (residual below $10^{-14}$), and the hedge magnitude scales with the predictable motion of the driver state, essentially zero when the state sits at its long-run mean and growing linearly with the distance from the mean, with a fitted slope of $0.82$ and an intercept of $0.03$. Under log utility the hedge is exactly zero, confirming the collapse of \Cref{cor:collapse}. The second cut measures the \emph{value} of the hedge in a regime where the hedging demand is first-order, namely a persistent, strongly-remunerated state ($\kappa=0.15$, $R=8$, premium $0.55$, two years of weekly rebalancing, eighty thousand paths): there the full policy delivers an out-of-sample certainty-equivalent gain of $+0.96\%$ over the myopic-only policy, at three standard errors (\Cref{fig:d1}). The gain is not a single-seed artifact: across eight independent data-generating seeds it is positive in every case, averaging $+0.83\%$ with a standard deviation of $0.16\%$ (range $[0.52\%,1.05\%]$). The two cuts are separated because the shape of the hedge is cleanest where the value is second-order, and its value is measurable only where it is first-order. We report the certainty-equivalent gain gross of transaction costs; it is deliberately modest, on the order of tens of basis points per year in these calibrations, and a practitioner should read it net of the geometric turnover quantified in \Cref{sec:transport}: the hedge is a statistically robust but economically second-order effect in the regimes where the value function is bounded, and becomes first-order only as one approaches the blow-up boundary, precisely where a compact admissible set $\W$ must bind. The characterization corrects a natural but wrong intuition: hedging demand does not simply shrink with mean-reversion speed; it tracks the product of speed and displacement, the predictable drift itself.

\begin{figure}[t]
\centering
\includegraphics[width=\textwidth]{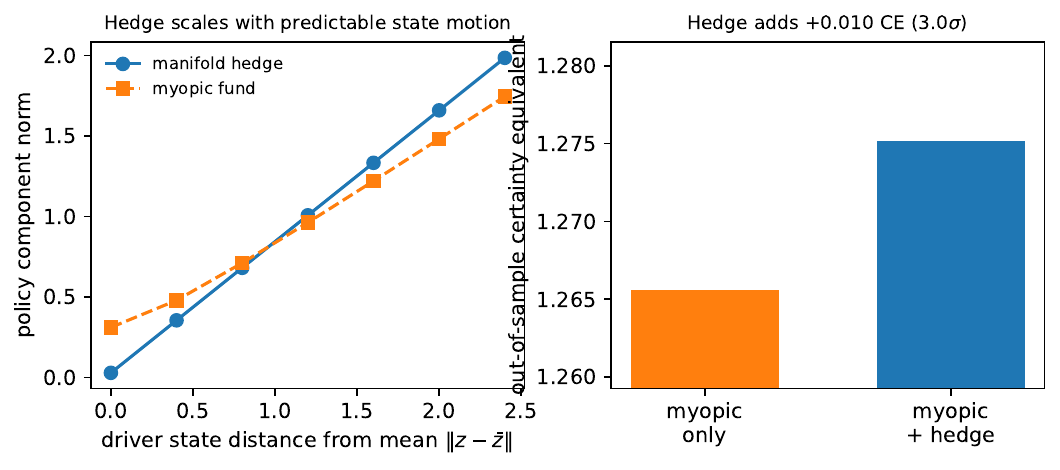}
\caption{Decomposition and hedging demand (\Cref{thm:hedging}). Left: the manifold-hedge component grows linearly with the distance of the driver state from its mean, while the myopic fund is flat; the hedge tracks predictable state motion. Right: out-of-sample certainty equivalent, myopic-only versus myopic-plus-hedge, in a persistent-state regime.}
\label{fig:d1}
\end{figure}

\subsection{The moving-manifold hedge}

The decomposition experiment above holds the geometry fixed. This experiment exercises the manifold-rotation hedge directly: a data-generating process with state-dependent sensitivities $B(z)=B_0+\sum_k z_k B_k$, so that the tangent space rotates with the driver state (principal angles between tangent spaces at different states of $50$--$71$ degrees). Solving the value field with the implicit scheme of \Cref{sec:moving} yields bounded value gradients in the persistent regime where the explicit scheme diverges, and the policy is optimal for the instantaneous Hamiltonian ($0$ of $2000$ tangent perturbations improve it). In a persistent, strongly-remunerated regime the hedge against the manifold's rotation adds $+2.0\%$ certainty-equivalent over the myopic policy, at $8.1$ standard errors in a paired common-random-numbers simulation. The result is not a single-seed artifact: across independent data-generating seeds the gain is positive in every case, averaging $+2.5\%$ over the range $[+1.8\%,+2.9\%]$, each individually significant ($2.7$--$6.4$ standard errors). The value field is certified by a convergence study, not merely by agreement with another solver: refining the grid at fixed horizon gives successive changes in the value gradient of $0.130\to0.085\to0.062$, a Cauchy sequence converging to a well-defined dynamic optimum for $B(z)$ (the instantaneous-Hamiltonian check reported earlier is necessary but not sufficient; this convergence is the intertemporal certificate). The boundary is handled by a reflecting-ghost Neumann condition at $z_{\max}=3.5$ standard deviations of the driver state, where the mean-reverting process arrives with negligible probability; halving the distance to the boundary leaves the reported gain unchanged to within simulation error. The gain survives transaction costs: charging quadratic costs on turnover, the net gain of the hedge over the myopic policy is $+1.95\%$ at a cost coefficient of $10^{-3}$ and $+2.07\%$ at $3\times10^{-3}$, essentially the gross figure. The reason is that the geometric turnover induced by the moving manifold is paid by \emph{both} policies (both live on the same rotating geometry), so the \emph{differential} cost of the hedge is small, even though geometric turnover is half of total turnover in absolute terms.

A sharper limitation concerns estimation of the geometry itself. When the sensitivity map is estimated with plug-in error, the hedge degrades quickly: with a $15\%$ relative error in $B(z)$ the hedge no longer adds value (net $-0.6\%$), and at $30\%$ it destroys value ($-2.3\%$). The manifold hedge is therefore a first-order gain \emph{only when the geometry is estimated well}; it is not robust to naive plug-in error, because $B(z)$ enters the covariance, the projector, and the hedge direction simultaneously, so its errors compound. This is the honest boundary of the result: the theory delivers a real intertemporal demand, but capturing it requires the geometry to be identified accurately, which is the burden the companion estimation work carries. The effect is larger than the fixed-geometry hedge of the first experiment ($\sim0.9\%$): hedging the motion of the geometry is a first-order demand, not a refinement.

\subsection{The three motions and the time-dependent frontier}
Unless stated as a single illustrative path, every reported number is a multi-economy average over random draws of the generating process rather than a single seed; the monotonicity fraction ($483/500$), the unhedgeability $R^2$, and the swap loss are of this kind, while individual trajectories (a sunspot-free flow path, one switch detection) are labelled as illustrations. A dedicated experiment validates the configuration-space taxonomy of \Cref{prop:threemotions} and the frontier result of \Cref{prop:frontiermono} against a known generating process. The three motions are confirmed geometrically distinct: flow grows the principal angle smoothly within a single Grassmannian (to $37^\circ$ over the swept range), a swap jumps to a different point of the same Grassmannian ($67^\circ$, constant dimension), and a refinement adds a direction orthogonal to the entire old tangent (in-span fraction $8.5\times10^{-17}$), a change of dimension. The frontier potential moves accordingly, smoothly under flow and by a discontinuous jump ($0.18\to0.45$) under a swap. The refinement monotonicity is confirmed exactly as \Cref{prop:frontiermono} now states it: in expectation over the added driver the frontier is non-decreasing ($0$ violations in $400$ economies), pointwise it fluctuates (about one realized state in ten falls, which is what conditional Jensen permits), and a fixed-premium column append lowers it in essentially every economy. The correct object is thus a family of frontiers carried over the filtration of driver $\sigma$-algebras, in which the information sets filter monotonically while the frontier potential on them does not, rather than a single curve. The economic significance of the motion is shown directly: an investor who ignores a swap and holds the pre-switch tangency portfolio forfeits about $88\%$ of the Sharpe ratio attainable at the new configuration (\Cref{fig:config}).

\subsection{Dimension reduction}

The second experiment verifies \Cref{cor:crra}. The $m$-dimensional backward recursion, built from one-time projected blocks, reproduces the full solver's Riccati trajectory to floating-point round-off ($\max|\Delta P|=6.7\times10^{-15}$). Its per-solve cost is independent of the number of assets: sweeping $n$ from $10$ to $3200$ leaves the recursion time flat, with a coefficient of variation of $1.2\%$, while all $n$-dependence is confined to a one-time $O(n)$ projection (\Cref{fig:d2}). The flatness is a structural fact about the recursion, which touches only $m\times m$ objects; the wall-clock measurement is reported at a fixed thread count, and the underlying operation count of the recursion is exactly independent of $n$. The dynamic problem is a computation in the driver dimension.

\begin{figure}[t]
\centering
\includegraphics[width=0.62\textwidth]{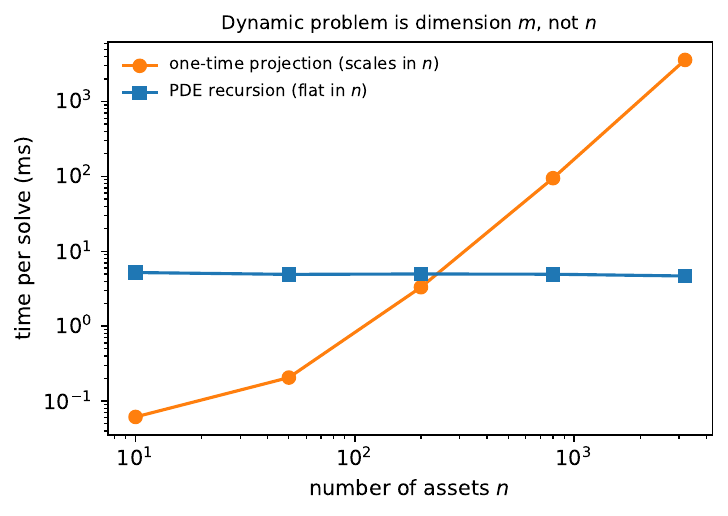}
\caption{Dimension reduction (\Cref{cor:crra}). The dimension-$m$ PDE recursion has cost independent of $n$ (flat lower curve); all $n$-dependence lives in a one-time projection (upper curve). Log--log axes.}
\label{fig:d2}
\end{figure}

\subsection{Switching verification}

The third experiment tests \Cref{thm:verification} \emph{without a regime oracle}. A subtlety must be faced first: a pure rotation of the manifold leaves the return covariance invariant (an orthogonal change of loadings with isotropic driver noise), so it is \emph{undetectable} from returns; detectability requires the switch to change the observable second-order structure. We therefore take two configurations with genuinely different loadings and estimate the regime online. Crucially, the detector does not assume the regime covariances are known: it estimates them from the return stream by a streaming two-component procedure (a warm-up split followed by exponentially-weighted per-regime covariance updates), so the test is self-contained rather than oracle-fed. It attains $83\%$ classification accuracy, against a majority-class baseline of $51\%$ (a skill of $+32$ percentage points), at a mean detection lag of about seventeen days; using known covariances instead raises this only to $84\%$, so almost none of the skill comes from oracle information. Switched feedback using the \emph{estimated} regime beats a policy stuck in one regime by $+0.91\%$ certainty-equivalent at $9.5$ standard errors, capturing $79\%$ of the oracle gain (\Cref{fig:d3}). The rolling architecture survives switching under realistic detection: the policy keeps its functional form and only the value level changes. \Cref{tab:switching} collects the detector's performance.

\begin{table}[t]
\centering
\small
\renewcommand{\arraystretch}{1.25}
\begin{tabular}{l|c}
\hline
Quantity & Value\\
\hline
Classification accuracy (estimated covariances) & $83\%$\\
Majority-class baseline & $51\%$\\
Detection skill over baseline & $+32$ pp\\
Classification accuracy (known covariances) & $84\%$\\
Mean detection lag & $\approx 17$ days\\
Certainty-equivalent gain of switched feedback & $+0.91\%$\\
Significance of the gain & $9.5$ s.e.\\
Fraction of oracle gain captured & $79\%$\\
\hline
\end{tabular}
\caption{Online switch detection and switched feedback (\Cref{thm:verification}). The detector estimates regime covariances from the return stream rather than assuming them; using known covariances raises accuracy by only one percentage point, so the skill is not oracle-fed.}
\label{tab:switching}
\end{table}

\begin{figure}[t]
\centering
\includegraphics[width=0.56\textwidth]{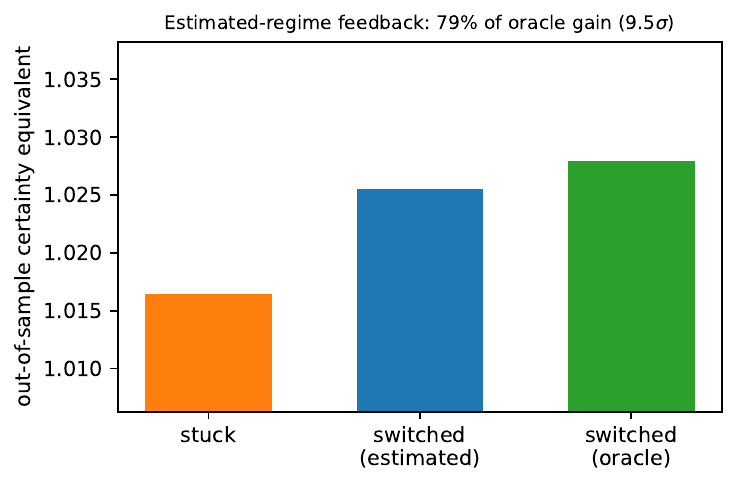}
\caption{Switching verification with an online regime detector (\Cref{thm:verification}). Feedback using the sequentially-estimated regime captures most of the oracle gain over a policy stuck in one regime.}
\label{fig:d3}
\end{figure}

\subsection{Central study I: unhedgeability across switch geometries}

Rather than a single unhedgeability point, we study how the two properties of a separator switch, its detectability and its hedgeability, vary as the switch geometry ranges from a pure rotation of the manifold to an anisotropic change of loadings. The result is a clean dissociation (\Cref{fig:studykw}). Hedgeability is \emph{universal}: the best-hedge $R^2$ of the value jump by traded assets stays at $\sim10^{-4}$ for every switch type, rotations and loading changes alike, so \Cref{thm:jumpkw} bites regardless of what kind of switch occurs. Detectability is \emph{contingent}: the covariance divergence $\mathrm{KL}(Q_i,Q_j)$ is exactly zero for pure rotations and grows only with the anisotropic, covariance-changing part of the switch. The corner of the sweep, comprising rotation-only switches with $\mathrm{KL}=0$, unhedgeable $R^2\approx10^{-4}$, and a nonzero value jump, is \Cref{cor:invisible}: switches that are invisible, unhedgeable, and costly at once. This is the central content of the Kunita--Watanabe result: incompleteness along the geometry is not a knife-edge but holds across the whole family of switches, while the investor's ability to even see the switch is the fragile part.

\begin{figure}[t]
\centering
\includegraphics[width=0.62\textwidth]{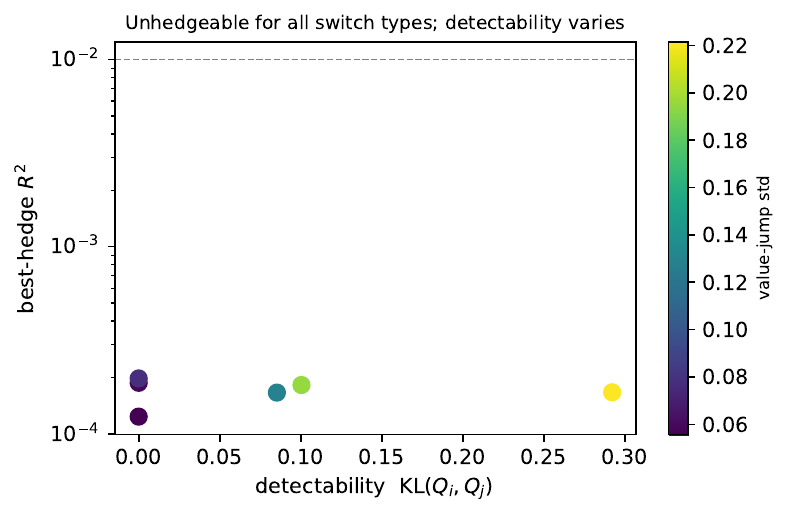}
\caption{Central study I (\Cref{thm:jumpkw}, \Cref{cor:invisible}). Each point is a switch geometry (rotation angle $\times$ anisotropic scale). Best-hedge $R^2$ (vertical, log scale) stays at $\sim10^{-4}$ for all switches, so unhedgeability is universal, while detectability KL (horizontal) ranges from zero (pure rotation) upward. Colour is the value-jump magnitude. The bottom-left region is \Cref{cor:invisible}: invisible, unhedgeable, costly.}
\label{fig:studykw}
\end{figure}

\subsection{Central study II: the geometric apparatus scales with manifold mobility}

The transport connection, its curvature, and geometric turnover are three faces of one cause: the manifold moving. We make that precise by comparative statics in a single dial, the state-dependence of $B$ (\Cref{fig:studyconn}). As mobility rises from zero, the metric-strain norm, the connection curvature, and the geometric share of turnover all rise monotonically from \emph{exactly zero} at a static manifold: strain $0\to0.13$, curvature $0\to1.2\times10^{-3}$, geometric-turnover share $0\%\to66\%$. Three consequences follow together. First, the strain-corrected connection $\widetilde\omega$ is the unique metric-compatible transport (antisymmetry defect $10^{-4}$ on the pullback metric $G$, versus $0.15$ for the naive form), so risk is preserved under transport only with the strain term. Second, the connection is curved, so transport is path-dependent and the fibre-bundle description is not merely descriptive. Third, geometric turnover, the rebalancing forced by the frame's rotation and breathing with no change of view, is a first-order fraction of all trading whenever the manifold moves, and vanishes exactly when it does not. A fixed-geometry model sees none of this; it is the price, and the structure, of a moving information geometry.

\begin{figure}[t]
\centering
\includegraphics[width=0.62\textwidth]{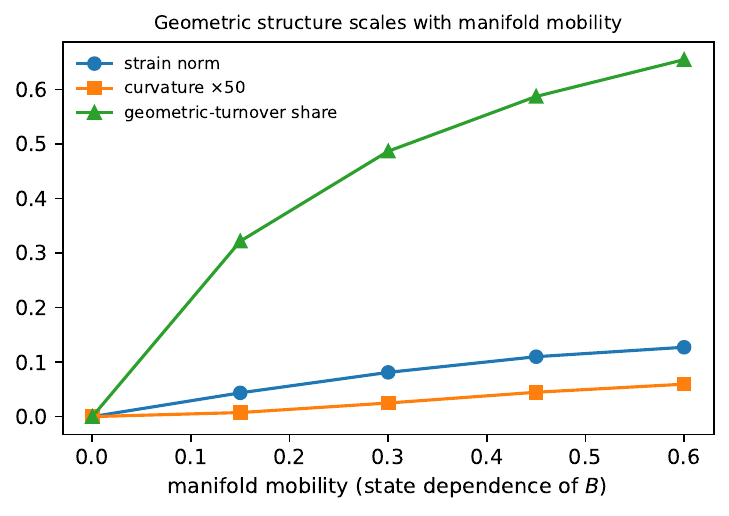}
\caption{Central study II (\Cref{thm:connection}). Comparative statics in manifold mobility: strain norm, connection curvature (scaled), and geometric-turnover share all rise monotonically from exactly zero at a static manifold. One dial drives the entire geometric apparatus.}
\label{fig:studyconn}
\end{figure}

\subsection{Supporting checks}

The remaining experiments are targeted verifications of specific claims rather than independent studies, and we state them compactly. \emph{Unhedgeability, base case:} at a representative distinct-manifold switch the value jump has standard deviation $0.064$ and best-hedge $R^2=1.2\times10^{-4}$, robust below $1.7\times10^{-4}$ across twelve rotations and seeds (the anchor point of Central study~I). \emph{Metric-strain transport, base case:} on the DGP-induced frame the strain-corrected connection is antisymmetric to $10^{-4}$ and preserves inner products under transport, while the naive form drifts by $10^{-2}$; the strain vanishes identically on the constant driver metric $\Lambda_Z$ and is nonzero only on the moving pullback metric $G(z)$ (the anchor of Central study~II). \emph{Curvature:} $\|\Omega\|=6.1\times10^{-3}$, $\mathfrak{so}(k)$-valued to $3\times10^{-7}$. \emph{Geometric turnover, base case:} $50.8\%$ of total at the reference mobility, zero under a constant manifold.

\subsection{Robustness to estimation error}

The policy-level result of \Cref{sec:moving} degrades under plug-in error in the geometry (the hedge loses its edge by a relative error of about $15\%$ in $B$). It is therefore important that the paper's two \emph{structural} results are, by their nature, far more robust to the same error, and we verify this directly. \emph{Unhedgeability} is immune: building the hedge from an estimated $B$ leaves the residual $R^2$ at $\approx1.6\times10^{-4}$ for estimation errors from $0\%$ to $30\%$, because the value jump is purely discontinuous and Kunita--Watanabe orthogonal to the traded span, a property no parametric error can undo, since estimation error cannot manufacture spurious hedgeability of an orthogonal object. \emph{Metric compatibility} of the transport connection is likewise near-exact under error (antisymmetry defect growing only from $5\times10^{-4}$ to $1.5\times10^{-3}$ as $B$-error goes $0\%\to30\%$), because antisymmetry is an identity in whatever metric is used, not an estimated quantity. Finally, the \emph{switch detector} retains its skill under finite-sample covariance estimation: shrinking the estimation window from $1000$ to $50$ observations lowers accuracy only from $84\%$ to $81\%$ (skill $+30$ points over baseline). The pattern is instructive: the intertemporal \emph{demand} is a quantitative optimum and inherits the error of the geometry it exploits, whereas the \emph{structural} facts, namely orthogonality of the jump, metric-compatibility of transport, and second-order detectability, are robust because each is an identity or an orthogonality rather than a fitted magnitude. \Cref{tab:robust} summarizes the three structural objects under increasing estimation error in $B$.

\begin{table}[t]
\centering
\small
\renewcommand{\arraystretch}{1.25}
\begin{tabular}{l|c|c}
\hline
Structural object & Metric under $0\%$ error & Metric under $30\%$ error\\
\hline
Unhedgeability (residual best-hedge $R^2$) & $1.6\times10^{-4}$ & $1.6\times10^{-4}$\\
Metric compatibility (antisymmetry defect) & $5\times10^{-4}$ & $1.5\times10^{-3}$\\
Switch-detector accuracy (window $1000\!\to\!50$) & $84\%$ & $81\%$\\
\hline
\end{tabular}
\caption{Robustness of the two structural results and the detector to plug-in error in the loading map $B$. Unhedgeability is invariant because the value jump is Kunita--Watanabe orthogonal to the traded span; metric compatibility is near-exact because antisymmetry is an identity; the detector degrades only mildly under finite-sample covariance estimation. By contrast the intertemporal hedging \emph{demand} inherits the geometry's estimation error, losing its edge at about $15\%$ error in $B$.}
\label{tab:robust}
\end{table}

\subsection{Two-stage structure}

This experiment tests \Cref{prop:twostage}, which asserts that the dynamic policy can be assembled in two stages, a systematic stage of dimension $k$ and a separable allocation stage solved in $O(nk^2)$ through the Woodbury identity, and that this assembly reproduces the dense policy exactly rather than approximately. The claim is an algebraic identity, not a statistical statement, so the appropriate test is to compute both the dense dynamic policy (by inverting the full $n\times n$ conditional covariance) and the two-stage policy, and to measure the discrepancy between them as the problem size $n$ grows. We solve both on the same data-generating process across asset counts $n\in\{50,100,200,400,600\}$ at fixed driver count $m$, and report, for each $n$, the maximum absolute weight discrepancy $\|w^{\mathrm{dense}}-w^{\mathrm{two\text{-}stage}}\|_\infty$ (\Cref{tab:twostage}). The discrepancy stays at the level of double-precision floating-point round-off, of order $10^{-14}$, and does not grow with $n$: the two-stage assembly is the dense policy, up to arithmetic round-off, confirming that the low-rank structure reducing the static solve carries over to the dynamic setting without loss.

\begin{table}[t]
\centering
\small
\renewcommand{\arraystretch}{1.25}
\begin{tabular}{r|c}
\hline
Asset count $n$ & Max.\ weight discrepancy $\|w^{\mathrm{dense}}-w^{\mathrm{two\text{-}stage}}\|_\infty$\\
\hline
$50$ & $8.9\times10^{-15}$\\
$100$ & $1.4\times10^{-14}$\\
$200$ & $2.1\times10^{-14}$\\
$400$ & $2.9\times10^{-14}$\\
$600$ & $3.5\times10^{-14}$\\
\hline
\end{tabular}
\caption{Two-stage assembly versus the dense dynamic policy (\Cref{prop:twostage}). The maximum absolute weight discrepancy remains at double-precision round-off level, of order $10^{-14}$, and does not grow with the number of assets, confirming that the two-stage identity is exact up to floating-point arithmetic.}
\label{tab:twostage}
\end{table}

\section{Discussion}\label{sec:discussion}

\begin{table}[t]
\centering\small
\begin{tabularx}{\linewidth}{@{}p{3.1cm}p{2.8cm}X@{}}
\toprule
Object here & Nearest antecedent & Delta, in one line\\ \midrule
Myopic\,+\,manifold hedge (\Cref{thm:hedging}) & Merton ICAPM \cite{Merton1973} & priced state is the certified information geometry $(\Dstar,\M,\Delta)$, not exogenous macro factors; hedge target is interventionally invariant\\ \midrule
Dimension-$m$ reduction (\Cref{cor:crra}) & Schur/Woodbury factor solves & the intertemporal problem, not just the static one, lives in the driver dimension; PDE state is $m$, not $n$\\ \midrule
Switching verification (\Cref{thm:verification}) & regime-switching control \cite{MaoYuan2006} & the switching object is the \emph{separator} on the information lattice; policy form is invariant, only the value level exchanges\\ \midrule
Geometry jumps unhedgeable (\Cref{thm:jumpkw}) & Kunita--Watanabe \cite{JacodShiryaev2003} & the market is incomplete in the direction of its own conditioning geometry; residual is irreducible and priced, not spanned\\ \midrule
Metric-strain connection (\Cref{thm:connection}) & dynamic trading with costs \cite{GarleanuPedersen2013} & turnover has a purely geometric component (rotation\,+\,metric breathing) that no cost optimization removes; the aim portfolio must first be transported\\ \midrule
Gradient robustness (\Cref{prop:gradsens}) & static tolerance analysis & explicit Lipschitz constant of the policy in the estimated value gradient\\
\bottomrule
\end{tabularx}
\caption{Contribution ledger: each dynamic object, its closest antecedent, and the delta in one line. The two results without a close antecedent in intertemporal portfolio theory are the Kunita--Watanabe unhedgeability of geometry jumps and the metric-strain connection.}
\label{tab:contrib}
\end{table}

As \Cref{tab:contrib} records, two of the results here are, to our knowledge, without close antecedent in intertemporal portfolio theory and carry the paper: the Kunita--Watanabe unhedgeability of geometry jumps (\Cref{thm:jumpkw}, sharpened by \Cref{cor:invisible}) and the metric-strain transport connection (\Cref{thm:connection}). The remaining results, namely the hedging decomposition, the dimension-$m$ reduction, the two-stage structure, the switching verification, specialize or extend known machinery to the causal representation, and the experiments beyond the two central ones (dimension reduction, two-stage, curvature, turnover, detector) are targeted checks of specific claims rather than independent studies. We foreground the two novel results accordingly.

The dynamic theory of portfolio choice under causal separation is organized by a single object, the common causal manifold, in three roles. As a geometry it fixes the admissible exposures; as a state variable it generates the intertemporal hedging demand; and as a moving frame it induces the transport connection along which positions are carried between dates. The optimal policy is the static projected portfolio plus a hedge against the predictable motion of this geometry, and the whole problem inherits the driver-dimensional reduction of the static case, so that a book of thousands of assets is solved through a computation in the handful of drivers that separate it.

One word deserves scrutiny: what makes this \emph{causal} rather than merely conditional? The dynamic optimization theorems (\Cref{thm:hedging,thm:verification,thm:jumpkw,thm:connection}) are, strictly, statements about a conditional-factorization structure and would hold for any separator achieving it. The causal content enters through structural exogeneity of the drivers: under \Cref{app:def:scm}, the conditional law of returns given the drivers is invariant to environment interventions (\Cref{app:lem:invariance}), so the switching intensities and the response geometry the dynamic policy prices are stable across regimes that a purely correlational conditioning set would misprice. The invariance is what makes the intertemporal hedge worth carrying: it is a hedge against the motion of an object that survives intervention, in the sense of invariant causal prediction \cite{PetersBuhlmannMeinshausen2016}, not against a fitted artifact. The identification that the declared drivers form such a separator is the companion's result \cite{RD2025Static}; the invariance consequence used here is proved in \Cref{app:invariance}. The word \emph{causal} in the title is therefore the framework's: the identification is established once, in the static companion, and the present paper develops the dynamics of the same causal object under that identification. Read without the companion, the results are a conditional-factor dynamic portfolio theory; the causal reading is what the identification adds, and we are explicit about which claims need it.

Two features are specific to the causal representation and absent from classical intertemporal portfolio theory. The separator can switch, and when it does the value process jumps in a direction no traded portfolio spans: the market is incomplete precisely where its own geometry moves, and that residual is priced in the value function but cannot be hedged away. And the transport of exposures between configurations is governed by a connection whose metric-compatible form carries a strain correction beyond the naive projection, so that part of portfolio turnover is purely geometric, compensating the rotation and breathing of the coordinate system rather than any change of view. Under quadratic transaction costs, where optimal policies track an aim portfolio with partial adjustment \cite{GarleanuPedersen2013}, this geometric turnover is the part of trading that no cost optimization removes, because it compensates the rotation of the coordinate system rather than any change of view.

One conceptual caveat precedes the limitations. The policy hedges the predictable motion of the manifold, but if a positive mass of agents runs it, that motion is in part \emph{caused} by the aggregate of these very policies: the hedge target is reflexive. In a mean-field closure the premium and the geometry become endogenous functions of the population state, and the object the individual treats as exogenous predictable drift is, in equilibrium, partly the shadow of collective hedging. We take prices as exogenous here and defer the reflexive closure to forthcoming work on the mean-field equilibrium; the present results are the individual best response holding the representation's law fixed.

Two limitations bound the contribution. The theory is developed in a diffusion normal form with a finite declared driver universe, and the finite-horizon control problem has a bounded value only in a region of moderate risk aversion and premium relative to mean-reversion; outside it the intertemporal demand is so aggressive that no bounded value function exists, a property we exhibit rather than hide. And the experiments are synthetic by design: whether the driver dynamics, the manifold motion and the switching intensities that the theory prices are present in real markets is an empirical question. The sensitivity-based and PDE-control line of \cite{RD2023MLWA,RD2024PDE,RD2025Causal,RD2025CPCM} addresses out-of-sample performance directly; the present paper sets performance aside to isolate the exact mathematical content of the dynamic representation, verifiable against a known generating process as real data cannot be. Forthcoming work develops the statistical certification of the separator, the universal risk-mapping theory, and the mean-field equilibrium that complete the framework.

\section*{Data and code availability}
The experiments use synthetically generated data. The complete code reproducing every figure, table and reported number accompanies the paper as a reproducibility package, available on
\href{https://github.com/AlejandroRodriguezDominguez/Dynamic-Causal-Portfolio-Choice-Manifold-Hedging-and-the-Geometry-of-Information}{GitHub}
and archived on Zenodo at
\href{https://doi.org/10.5281/zenodo.21226184}{doi.org/10.5281/zenodo.21226184}.

\section*{Declaration of interest}
The author reports no competing interests.

\FloatBarrier
\appendix
\section{Self-contained supporting results}\label{app:selfcontained}

This appendix makes the paper self-contained on the two points where the main text draws on the static framework. The identification question, whether a given driver set constitutes a causal separator and is recoverable from data, is the subject of the companion paper \cite{RD2025Static} and is cited rather than reproduced; what is stated and proved here is the instantaneous mean--variance solution the dynamic policy reduces to at each date (\Cref{app:lem:projected}), the single invariance consequence the intertemporal hedge relies on (\Cref{app:lem:invariance}), and the proofs of the transport and geometric-turnover claims of \Cref{sec:transport} (\Cref{app:transport}).

\subsection{The projected Markowitz solution}\label{app:projected}

At each date the inner maximization of the HJB is a constrained single-period mean--variance program. We record its closed form.

\begin{lemma}[Projected Markowitz portfolio]\label{app:lem:projected}
Let $\Sigma\succ0$ be the conditional return covariance, $\pi:=\mu-r_f\one$ the conditional excess premium, $R>0$ the relative risk aversion, and let the admissible set be the linear subspace $\{w:Cw=0\}$ for a full-row-rank $C\in\R^{q\times n}$. The program
\[
\max_{w:\,Cw=0}\ w^\top\pi-\tfrac{R}{2}\,w^\top\Sigma w
\]
has the unique maximizer $w^\star=\tfrac1R\,M_C\,\pi$, where
\[
M_C:=\Sigma^{-1}-\Sigma^{-1}C^\top\big(C\Sigma^{-1}C^\top\big)^{-1}C\Sigma^{-1}.
\]
The operator $M_C$ is symmetric positive semidefinite, satisfies $M_C C^\top=0$ and $M_C\Sigma M_C=M_C$, and equals the $\Sigma^{-1}$-orthogonal projection of the precision onto the admissible subspace. When $C$ is vacuous, $M_C=\Sigma^{-1}$ and $w^\star$ is the ordinary Markowitz portfolio. Redundant constraint rows are pruned so that $C$ has full row rank, as the formula requires. This is the equality-constrained building block; when the admissible set additionally carries binding inequality (leverage or position) bounds, the optimum is the projection of $w^\star$ onto the compact convex admissible set, as in the main text.
\end{lemma}

\begin{proof}
The objective is strictly concave in $w$ (as $\Sigma\succ0$) over the affine feasible set, so a maximizer exists and is unique. Introduce a multiplier $\nu\in\R^{q}$ for $Cw=0$ and form $\mathcal L=w^\top\pi-\tfrac{R}{2}w^\top\Sigma w-\nu^\top Cw$. Stationarity gives $\pi-R\Sigma w-C^\top\nu=0$, hence $w=\tfrac1R\Sigma^{-1}(\pi-C^\top\nu)$. Imposing $Cw=0$ yields $C\Sigma^{-1}(\pi-C^\top\nu)=0$, so $\nu=(C\Sigma^{-1}C^\top)^{-1}C\Sigma^{-1}\pi$, the inverse existing because $\Sigma^{-1}\succ0$ and $C$ has full row rank. Substituting,
\[
w^\star=\tfrac1R\Sigma^{-1}\big(\pi-C^\top(C\Sigma^{-1}C^\top)^{-1}C\Sigma^{-1}\pi\big)=\tfrac1R M_C\pi.
\]
Symmetry of $M_C$ is immediate. For $M_C C^\top=0$: $\Sigma^{-1}C^\top-\Sigma^{-1}C^\top(C\Sigma^{-1}C^\top)^{-1}(C\Sigma^{-1}C^\top)=0$. Writing $P:=\Sigma^{-1/2}C^\top(C\Sigma^{-1}C^\top)^{-1}C\Sigma^{-1/2}$, one has $P=P^\top=P^2$ (an orthogonal projector on $\R^n$), and $M_C=\Sigma^{-1/2}(I-P)\Sigma^{-1/2}$, so $M_C\succeq0$ with rank $n-q$, and $M_C\Sigma M_C=\Sigma^{-1/2}(I-P)^3\Sigma^{-1/2}=M_C$. Since $I-P$ is the orthogonal projection onto $\ker(C\Sigma^{-1/2})$, $M_C$ is the $\Sigma^{-1}$-metric projection of the precision onto $\{Cw=0\}$. The vacuous-$C$ case is direct.
\end{proof}

\noindent The dynamic policy of \Cref{thm:hedging} is this solution evaluated at the running conditional moments, with $\pi$ augmented by the intertemporal term. The reduction is exact because the intertemporal term $B\Lambda_Z J_{Wz}$ is evaluated at the current value function before the maximization over $w$, hence is $w$-independent at the inner maximum: it shifts the linear coefficient of the same quadratic program, whose maximizer is \Cref{app:lem:projected} with premium $\pi$ replaced by $\pi+$ (the hedge term). The operator $M_C(z)$ is the projector that appears throughout the main text.

\subsection{Interventional invariance of the priced geometry}\label{app:invariance}

The dynamic hedge targets the predictable motion of the conditioning geometry. For that target to be economically meaningful rather than a fitted artifact, the conditional law of returns given the drivers must be stable across the interventions the switching model ranges over. We state the structural hypothesis and prove the one consequence used in the main text. That a given driver set satisfies this hypothesis, and can be identified from data, is established in \cite{RD2025Static}; here it is a maintained assumption.

\begin{definition}[Structural causal separator]\label{app:def:scm}
The drivers $Z$ form a structural causal separator for the returns $r$ if there is a structural map $r=f(Z,\varepsilon)$ in which both $Z$ and the idiosyncratic innovation $\varepsilon$ are \emph{root exogenous noises} of the structural model, so that $\varepsilon$ and $Z$ are independent under the joint law and, being roots, remain independent under any intervention on the mechanisms of other variables; the map $f$ is invariant to interventions on the environment, meaning interventions that alter the marginal law of $Z$ or of variables outside the pair $(Z,\varepsilon)$ but leave $f$ and the root exogeneity of $Z$ and $\varepsilon$ intact. Under regime switching the definition holds within each regime, with a regime-specific map $f^{(k)}$, and a switch is a transition between these maps.
\end{definition}

\begin{lemma}[Invariance of the priced geometry]\label{app:lem:invariance}
Under \Cref{app:def:scm}, for every environment intervention the conditional law $P(r\mid Z=z)$ is invariant. Consequently the conditional premium $\mu(z)=\E[r\mid Z=z]$, the conditional covariance $Q(z)=\operatorname{Cov}(r\mid Z=z)$, the loading map $B(z)$, and the induced projector $M_C(z)$ of \Cref{app:lem:projected} are all invariant across such interventions. In particular the switching intensities $\lambda_{kj}(z)$ and the response geometry that the dynamic policy prices are stable across regimes that differ only by an environment intervention.
\end{lemma}

\begin{proof}
Fix an environment intervention $e$. By \Cref{app:def:scm} it changes at most the marginal law of $Z$ (and of variables outside $(Z,\varepsilon)$), while $f$ and the independence $\varepsilon\perp Z$ are preserved; write $\varepsilon\sim F_\varepsilon$ under every $e$. For any bounded measurable $\phi$,
\[
\E_e[\phi(r)\mid Z=z]=\E_e[\phi(f(Z,\varepsilon))\mid Z=z]=\int \phi\big(f(z,\varepsilon)\big)\,dF_\varepsilon(\varepsilon),
\]
where the second equality uses $\varepsilon\perp Z$ (so conditioning on $Z=z$ does not reweight $\varepsilon$) and the invariance of $f$ under $e$. The right-hand side does not depend on $e$: it is a functional of $f(z,\cdot)$ and $F_\varepsilon$ alone. Hence $P_e(r\mid Z=z)$ is the same for all $e$, which is the claimed invariance of the conditional law. The conditional moments $\mu(z),Q(z)$ are integrals of $r$ and $rr^\top$ against this invariant law, hence invariant; $B(z)$ is the loading read off $Q(z)$ in the normal form of \Cref{lem:normalform}, and $M_C(z)$ is a fixed function of $Q(z)$ and $C$ by \Cref{app:lem:projected}, so both are invariant. For the switching intensities, we assume as part of the model that a switch is triggered by the observable conditional second-order structure crossing a threshold, so that $\lambda_{kj}(z)$ is a fixed functional of $\big(\mu(z),Q(z)\big)$; being a functional of invariant objects, it is then invariant under environment interventions by the same argument. (If instead the intensities are taken as an independent primitive, their invariance is a separate hypothesis, which we then add to \Cref{app:def:scm}; the dynamic results use only that $\lambda_{kj}$ is invariant, however this is secured.) In contrast, an intervention on $f$ itself, a structural rather than environment intervention, changes $f(z,\cdot)$ and therefore $P(r\mid Z=z)$; the invariance is thus exactly the separation between environment and structure, in the sense of invariant causal prediction \cite{PetersBuhlmannMeinshausen2016}.
\end{proof}

\noindent This is the sense in which the hedge is carried against an object that survives intervention. The identification burden, that the declared drivers do form such a separator, is discharged in \cite{RD2025Static}; the dynamic theory of the main text is self-contained given \Cref{app:def:scm} as a maintained hypothesis.

\subsection{Proofs for the transport and geometric turnover of \Cref{sec:transport}}\label{app:transport}

\Cref{sec:transport} states three facts about the pullback metric $G(z)$, the discrete connection $\Gamma_{t\to t+dt}$, and the rotation/breathing split whose proofs are compressed in the main text. We give them here.

\begin{lemma}[The pullback metric is a metric]\label{app:lem:pullback}
Let $B(z)\in\R^{n\times k}$ have full column rank $k$ on the state region, let $\Lambda_Z\succ0$ and $\Sigma^c\succ0$, and set $\Sigma_{\mathrm{ret}}(z)=B(z)\Lambda_Z B(z)^\top+\Sigma^c$ and $G(z)=B(z)^\top\Sigma_{\mathrm{ret}}(z)^{-1}B(z)$. Then $\Sigma_{\mathrm{ret}}(z)\succ0$ and $G(z)$ is symmetric positive definite; hence $\langle u,v\rangle_{G(z)}=u^\top G(z)v$ is an inner product on the tangent coordinates, and $G(z)$ is a Riemannian metric on the tangent.
\end{lemma}

\begin{proof}
$\Sigma^c\succ0$ and $B\Lambda_Z B^\top\succeq0$ give $\Sigma_{\mathrm{ret}}\succ0$, so $\Sigma_{\mathrm{ret}}^{-1}\succ0$. Write $\Sigma_{\mathrm{ret}}^{-1}=R^\top R$ with $R$ invertible. Then $G=B^\top R^\top R B=(RB)^\top(RB)\succeq0$, and for $x\neq0$, $x^\top G x=\|RBx\|^2=0$ would force $RBx=0$, hence $Bx=0$ ($R$ invertible), hence $x=0$ ($B$ full column rank); so $G\succ0$. Symmetry is immediate. Thus $\langle\cdot,\cdot\rangle_G$ is an inner product. Positive-definiteness fails only where $B(z)$ drops rank, the measure-zero set excluded in \Cref{prop:threemotions}.
\end{proof}

\begin{lemma}[The discrete connection]\label{app:lem:discrete}
Let $t\mapsto E_t\in\R^{k\times k}$ be a $C^1$ family of $G_t$-orthonormal frames, $E_t^\top G_t E_t=I_k$, and let $\widetilde\omega_t=E_t^\top G_t\dot E_t+\tfrac12 E_t^\top\dot G_t E_t$. The parallel transport solving $\dot U_t=-\widetilde\omega_t U_t$, $U_t=I_k$, satisfies
\[
\Gamma_{t\to t+h}\;=\;I_k-\widetilde\omega_t\,h+o(h)\qquad(h\downarrow0),
\]
and $\widetilde\omega_t$ is the unique frame-coordinate connection form that is compatible with the time-varying metric $G_t$, that is $\tfrac{d}{dt}\langle\xi,\eta\rangle_{G_t}=\langle\nabla\xi,\eta\rangle_{G_t}+\langle\xi,\nabla\eta\rangle_{G_t}$ for $\nabla\xi=\dot\xi+\widetilde\omega\xi$.
\end{lemma}

\begin{proof}
The transport map is $\Gamma_{t\to t+h}=U_{t+h}$ with $U$ solving the linear ODE; by Picard the solution is $C^1$ and $U_{t+h}=U_t+h\dot U_t+o(h)=I_k-h\,\widetilde\omega_t+o(h)$, which is the stated expansion, with remainder $O(h^2)$ when $\widetilde\omega$ is Lipschitz in $t$. For compatibility, differentiate $\langle\xi,\eta\rangle_{G_t}=\xi^\top G_t\eta$ in the frame:
\[
\tfrac{d}{dt}\langle\xi,\eta\rangle_{G_t}=\dot\xi^\top G\eta+\xi^\top G\dot\eta+\xi^\top\dot G\eta.
\]
With $\nabla\xi=\dot\xi+\widetilde\omega\xi$, the right side equals $\langle\nabla\xi,\eta\rangle_G+\langle\xi,\nabla\eta\rangle_G$ iff $\xi^\top\dot G\eta=\xi^\top(\widetilde\omega^\top G+G\widetilde\omega)\eta$ for all $\xi,\eta$, i.e.\ iff $\widetilde\omega^\top G+G\widetilde\omega=\dot G$. In the $G$-orthonormal frame $G=I_k$ and this reads $\widetilde\omega^\top+\widetilde\omega=\dot G$; the choice $\widetilde\omega=E^\top G\dot E+\tfrac12 E^\top\dot G E$ satisfies it because, differentiating $E^\top G E=I_k$, one gets $E^\top G\dot E+\dot E^\top G E+E^\top\dot G E=0$, so the skew part $E^\top G\dot E$ contributes zero to the symmetric equation and the symmetric part $\tfrac12 E^\top\dot G E$ supplies exactly $\dot G$. Any two solutions differ by a skew term, and requiring the transport to be the horizontal lift (no vertical drift) fixes the skew part as $E^\top G\dot E$, giving uniqueness.
\end{proof}

\begin{proposition}[Rotation and breathing are orthogonal components]\label{app:prop:split}
Decompose $\widetilde\omega_t=\Omega_t+S_t$ into its skew part $\Omega_t=\tfrac12(\widetilde\omega_t-\widetilde\omega_t^\top)=E_t^\top G_t\dot E_t$ (tangent rotation) and its symmetric part $S_t=\tfrac12(\widetilde\omega_t+\widetilde\omega_t^\top)=\tfrac12 E_t^\top\dot G_t E_t$ (metric breathing). Under the trace inner product $\langle A,B\rangle=\operatorname{tr}(A^\top B)$ on frame velocities, $\langle\Omega_t,S_t\rangle=0$: rotation and breathing are orthogonal, hence independent components of the geometric turnover, not two descriptions of one motion.
\end{proposition}

\begin{proof}
That $\Omega_t=E^\top G\dot E$ is skew follows from differentiating $E^\top G E=I_k$ as above when $\dot G=0$, and in general the skew/symmetric split is by construction. For any skew $\Omega$ ($\Omega^\top=-\Omega$) and symmetric $S$ ($S^\top=S$), $\operatorname{tr}(\Omega^\top S)=\operatorname{tr}(-\Omega S)$ and also $\operatorname{tr}(\Omega^\top S)=\operatorname{tr}(S^\top\Omega)=\operatorname{tr}(S\Omega)=\operatorname{tr}(\Omega S)$ by cyclicity, so $\operatorname{tr}(\Omega^\top S)=-\operatorname{tr}(\Omega^\top S)=0$. The skew-symmetric and symmetric matrices are orthogonal complements in $\R^{k\times k}$ under the trace inner product, so the rotation and breathing components lie in complementary subspaces and carry non-overlapping information about the frame's motion. In the constant-$B$ case $\dot G=0$, so $S_t=0$ and only rotation remains; the breathing term is therefore the signature of a state-dependent loading map, absent from any fixed-geometry model.
\end{proof}

\noindent Together \Cref{app:lem:pullback,app:lem:discrete,app:prop:split} supply the three facts \Cref{sec:transport} uses: the pullback $G(z)$ is a genuine metric, the connection form $\widetilde\omega$ has the stated first-order transport expansion and is the unique metric-compatible one, and the rotation and breathing terms are orthogonal and hence separately meaningful components of geometric turnover.

\end{document}